%% file: TCNS_Katz_centrality_mul.tex
\documentclass[journal,twoside,web]{ieeecolor}
\usepackage{generic}
\usepackage{cite}
\usepackage[hidelinks]{hyperref}
\usepackage{amsmath,amssymb,amsfonts}
\usepackage{algpseudocode,algorithm}
\usepackage{graphicx}
\usepackage{textcomp}
\usepackage{subcaption}
\usepackage{xcolor}
\usepackage{tabularx}
\usepackage{academicons}
\usepackage{soul}
\usepackage{multirow}
\usepackage{diagbox}
\usepackage[colorinlistoftodos]{todonotes}
\usetikzlibrary{arrows,shapes.geometric,positioning}
\usepackage{mathtools}
\mathtoolsset{showonlyrefs}
\input{deflatex3}

\newcommand{\argmax}{\operatorname*{arg\,max}}
\newtheorem{Lemma}{Lemma}
\newtheorem{Theorem}{Theorem}
\newtheorem{Definition}{Definition}
\newtheorem{Assumption}{Assumption}
\newtheorem{Remark}{Remark}
\newtheorem{Problem}{Problem}
\newtheorem{Proposition}{Proposition}

\usepackage{svg}
\newcommand{\orcid}[1]{\href{https://orcid.org/#1}{\includesvg[width=8pt]{orcid}}}

\makeatletter
\newcommand{\multiline}[1]{%
  \begin{tabularx}{\linewidth-\ALG@thistlm-0.0cm}[t]{@{}X@{}}
    #1
  \end{tabularx}
}
\makeatother
\def\BibTeX{{\rm B\kern-.05em{\sc i\kern-.025em b}\kern-.08em
T\kern-.1667em\lower.7ex\hbox{E}\kern-.125emX}}
\begin{document}
\title{Katz Centrality-Based Security Allocation in Positive Networks}
\author{Anh Tung Nguyen, Sribalaji C. Anand, André M. H. Teixeira 
\thanks{This work is supported by the Swedish Research Council under the grant 2024-00185, by the Swedish Foundation for Strategic Research, and by the Knut and Alice Wallenberg Foundation.}
\thanks{Anh Tung Nguyen and Andr{\'e} M. H. Teixeira are with the Department of Information Technology, Uppsala University, PO Box 337, SE-75105, Uppsala, Sweden (e-mail: \{anh.tung.nguyen, andre.teixeira\}@it.uu.se).}
\thanks{ Sribalaji C. Anand is with the Department of Electrical and Systems Engineering, University of Pennsylvania, United States. He is also affiliated with the School of Electrical Engineering and Computer Science, KTH Royal Institute of Technology, Stockholm, Sweden (e-mail: srca@kth.se). }
}
	
\maketitle

\begin{abstract}                         
This paper deals with security allocation challenges for networked control systems represented by positive-weighted digraphs under stealthy false data injection attacks. These systems consist of interconnected subsystems, referred to as nodes in the underlying digraph, where an adversary aims to maximize network performance loss by stealthily attacking specific nodes.
Meanwhile, a defender monitors several nodes to impose stealthiness constraints on the adversary's actions, thereby minimizing the network performance loss. We analyze the worst-case network performance loss of these stealthy attacks and make the following contributions: we (i) show that the worst-case network performance loss is upper-bounded by a tractable semi-definite programming (SDP) problem; (ii) establish the relationship between the SDP problem and the Katz centrality measure of the underlying digraph under a sufficient condition, resulting in a network-size-independent optimization problem; and (iii) provide a heuristic search based on the Katz centrality measure of the underlying digraph for selecting sub-optimal monitor nodes against all admissible attack scenarios without solving optimization problems. These results offer practical insights for safeguarding large-scale networked control systems against stealthy false data injection attacks. The obtained results are validated via extensive simulations on Erdős–Rényi random graphs with different network sizes.
\end{abstract}
\section{Introduction}
Control systems are deeply integrated into various critical infrastructures, such as power grids, transportation systems, and water distribution networks, many of which can be modeled or approximated as positive networked systems \cite{molzahn2017survey,eliades2024smart,conti2021survey}. Due to their large scale and societal importance, these systems are often divided into smaller interconnected parts to be managed efficiently. 
These parts often rely on open communication technologies to share their operating information, including public Internet and wireless networks, which leave them vulnerable to cyber threats \cite{falliere2011w32, Havex2014, kshetri2017hacking, Israel2020, miller2021looking}. The potential consequences of such vulnerabilities are both significant and far-reaching, impacting finances and public safety. Notable examples include the devastating effects of the Stuxnet malware on an Iranian industrial system in 2010 \cite{falliere2011w32}, the Havex Trojan malware on European infrastructure utilities in 2014 \cite{Havex2014}, the Industroyer attack on Ukraine’s power grid in 2016 \cite{kshetri2017hacking}, and the thwarted Triton-like malware on an Israeli water distribution network in 2020 \cite{Israel2020}. More reported attacks can be found in \cite{miller2021looking}.
As a result, ensuring the security of control systems has become a matter of critical importance. 

In an attempt to handle cyber attacks on control systems, various recent studies \cite{teixeira2015strategic,griffioen2024ensuring,murguia2020security,pirani2021game} put effort into studying security metrics that consider policies and purposes of defenders and adversaries simultaneously. Tractable security metrics can be directly applied to deal with the sensor placement problem \cite{nguyen2024security,li2018false,yuan2019stackelberg,umsonst2021bayesian, milovsevic2023strategic,dahan2022network} and the actuator protection problem \cite{anand2023risk,shukla2022robust}. More specifically, \cite{li2018false,yuan2019stackelberg,nguyen2024security} cast the security resource allocation problem within the Stackelberg game framework, aiming to identify the optimal actions for the defender and the adversary by solving a single optimization problem. This approach incorporates all admissible adversarial actions via a simplified linear mapping from action spaces to game payoffs. The studies \cite{shukla2022robust,milovsevic2023strategic,dahan2022network} adopt a different game-theoretic approach wherein both the defender and the adversary operate within discrete action spaces. This setup facilitates the application of traditional backward induction to compute the Stackelberg equilibrium \cite{shukla2022robust} or the bilevel optimization for finding the zero-sum game equilibrium \cite{milovsevic2023strategic,dahan2022network}, albeit at the cost of substantial computational complexity due to the need to evaluate all admissible actions of the players, as noted in \cite[Section V]{shukla2022robust}.

A notable gap in the existing literature \cite{li2018false,yuan2019stackelberg,shukla2022robust,umsonst2021bayesian,milovsevic2023strategic,dahan2022network} is the lack of analysis and computation of the performance loss caused by stealthy attacks on control systems. 
This issue has been partially addressed in \cite{anand2023risk, pirani2021game, pirani2020graph, pirani2021strategic,weerakkody2016graph,zhang2023structural,nguyen2024security}. Specifically, the works \cite{weerakkody2016graph,zhang2023structural} mainly study the perfectly undetectable attacks. The trade-off between the malicious impact and the detection energy is studied in \cite{pirani2021game, pirani2020graph, pirani2021strategic}. 
However, the metrics proposed in \cite{weerakkody2016graph,zhang2023structural, pirani2021game, pirani2020graph, pirani2021strategic} are not directly applicable to the study of stealthy attacks due to the presence of stealthiness constraints, which are of significant practical importance.
Our previous works \cite{anand2023risk, anand2024risk, nguyen2024security} address this gap by examining the worst-case malicious impact under stealthiness constraints through the $\Lc_2$ strategic output-to-output gain security metric \cite{teixeira2015strategic}. However, this metric poses a computational challenge when applied to large-scale networks.

The computational challenge in large-scale networks
has been addressed in our previous study \cite{nguyen2025scalable} by adopting the results of Kalman-Yakubovich-Popov (KYP) Lemma in positive systems \cite{rantzer2015kalman}. While this approach enhances computational efficiency, it lacks a graph-theoretic interpretation, e.g., set packing and set cover in \cite{milovsevic2023strategic,dahan2022network}.
On the other hand, our previous work \cite{nguyen2024security} has already provided a binary graph-theoretic answer to the security evaluation question based on dominating sets, determining whether the security metric is finite or infinite. However, a fundamental graph-theoretic property of the output-to-output gain security metric, which would enable scalable security evaluation, remains unexplored. Consequently, this paper investigates a graph-theoretic characterization of the security metric and its implications for scalable security assessment and heuristic monitor allocation.

In this paper, we consider a networked control system, associated with a positive-weighted digraph, under stealthy attacks. 
The system states are constrained to be non-negative, which reflects many real-world networks, such as transportation and water distribution networks.
The system consists of interconnected one-dimensional subsystems, known as nodes in the digraph. The adversary aims to maximally degrade the network performance  
by selecting several specific nodes to carry out stealthy false data injection attacks, targeting the information transmitted from these nodes to their neighbors. In contrast, a defender monitors the outputs of some nodes, which imposes a stealthiness condition on the adversary's actions, aiming to minimize the network performance loss. The security problem outlined above is illustrated in Fig.~\ref{fig:problem}. We study the worst-case performance loss of stealthy attacks (also henceforth referred to
as the worst-case disruption) and provide the following contributions:
\begin{itemize}
    \item The worst-case disruption is upper-bounded by a tractable SDP problem using dissipative systems theory \cite{moylan2014dissipative}. 
    \item By exploiting the positivity and holding a sufficient robustness condition, we show that the worst-case disruption exactly equals the solution to the SDP problem. Furthermore, the SDP problem is simplified into an optimization problem independent of network size, making it well-suited for large-scale networks.
    \item We introduce a necessary and sufficient geometric condition based on two adapted versions of the Katz centrality measure \cite{katz1953centrality} under which the simplified SDP problem admits a finite solution.
    \item By leveraging the above geometric condition, we provide a heuristic search based on the two adapted Katz centrality measures to find sub-optimal monitor nodes against all the admissible attack scenarios, bypassing the need for solving a minimization of the worst-case disruption.
\end{itemize}

The remainder of this paper is organized as follows.
Section \ref{sec:problem} describes a positive networked control system under stealthy attacks and the worst-case disruption.
Thereafter, Section~\ref{sec:robust_security} presents the computation of the worst-case disruption in an SDP problem. This SDP problem can then be adapted to suit large networks under a sufficient robustness condition. 
In Section~\ref{sec:graph_security}, the evaluation of the worst-case disruption is carried out by examining two adapted versions of the Katz centrality measure, leading to a geometric characterization and a heuristic search for candidate monitor nodes. 
Section~\ref{sec:simulation} validates the obtained results on Erdős–Rényi random graphs while 
Section~\ref{sec:conclusion} concludes the paper. We conclude this section by providing the notation to be used throughout this paper.

\textbf{Notation:} $\Rbb^n  (\Rbb^n_{>0},\Rbb^n_{\geq 0})$ and $\Rbb^{n \times m}$ stand for sets of real (positive, non-negative) $n$-dimensional vectors and real $n$-by-$m$ matrices, respectively; 
the set of $n$-by-$n$ symmetric positive semidefinite matrices is denoted as $\Sbb^n_{\succeq 0}$.
We denote $A \succ(\succeq) B$ if $A-B$ is a positive definite (semi-definite) matrix.
An $i$-th column of the $n$-by-$n$ identity matrix is denoted as $e_i$.
For a set $\Ac$, $|\Ac|$ stands for the set cardinality. 
The space of square-integrable  functions is defined as $\Lc_{2} \triangleq \bigl\{f: \Rbb_{>0} \rightarrow \Rbb^n | \norm{f}^2_{\Lc_2 [0,\infty]} < \infty \bigr\}$
and the extended space is defined as $\Lc_{2e} \triangleq \bigl\{ f: \Rbb_{>0} \rightarrow \Rbb^n | \norm{f}^2_{\Lc_2 [0,H]} < \infty, \forall 0 < H < \infty \bigr\} $ where $\norm{f}_{\Lc_2 [0,H]}^2 \triangleq \int_{0}^{H} \norm{f(t)}_2^2 \, \text{d}t$.
The notation $\norm{f}^2_{\Lc_2}$  is used  as shorthand for the norm $\norm{f}_{\Lc_2 [0,H]}^2$ if the time horizon $[0,H]$ is clear from the context.
Let $\Gc \triangleq (\Vc, \Ec, A, \Theta)$ be a digraph with the set of $N$ nodes $\Vc = \{1, 2,...,N\}$, the set of edges $\Ec \subseteq \Vc \times \Vc $, the adjacency matrix $A = [A_{ij}]$, and self-loop gain diagonal matrix $\Theta = \textbf{diag}([\theta_i])$ where \textbf{diag} stands for a diagonal matrix.
Each element $(i,j) \in \Ec, i\neq j$, represents a directed edge from $j$ to $i$ and 
the element $A_{ij}$ of the adjacency matrix is positive, and with $(i,j) \notin \Ec$ or $i = j$, $A_{ij} = 0$. 
The in-degree Laplacian matrix is defined as $L = [\ell_{ij}] = D_{\rm in} - A$ where $D_{\rm in} = \textbf{diag}(A \textbf{1}) + \Theta$ and $\textbf{1}$ is an all-ones vector.
Further, $\Gc$ is called a strongly connected digraph if and only if $\sum_{k = 1}^{N-1} A^k$ has no zero entry.
The set of all in-neighbors of $i$ is denoted as $\Nc_i = \{j \in \Vc| (i,j) \in \Ec \}$.

\begin{figure}[!t]
    \centering
    \includegraphics[width=0.9\linewidth]{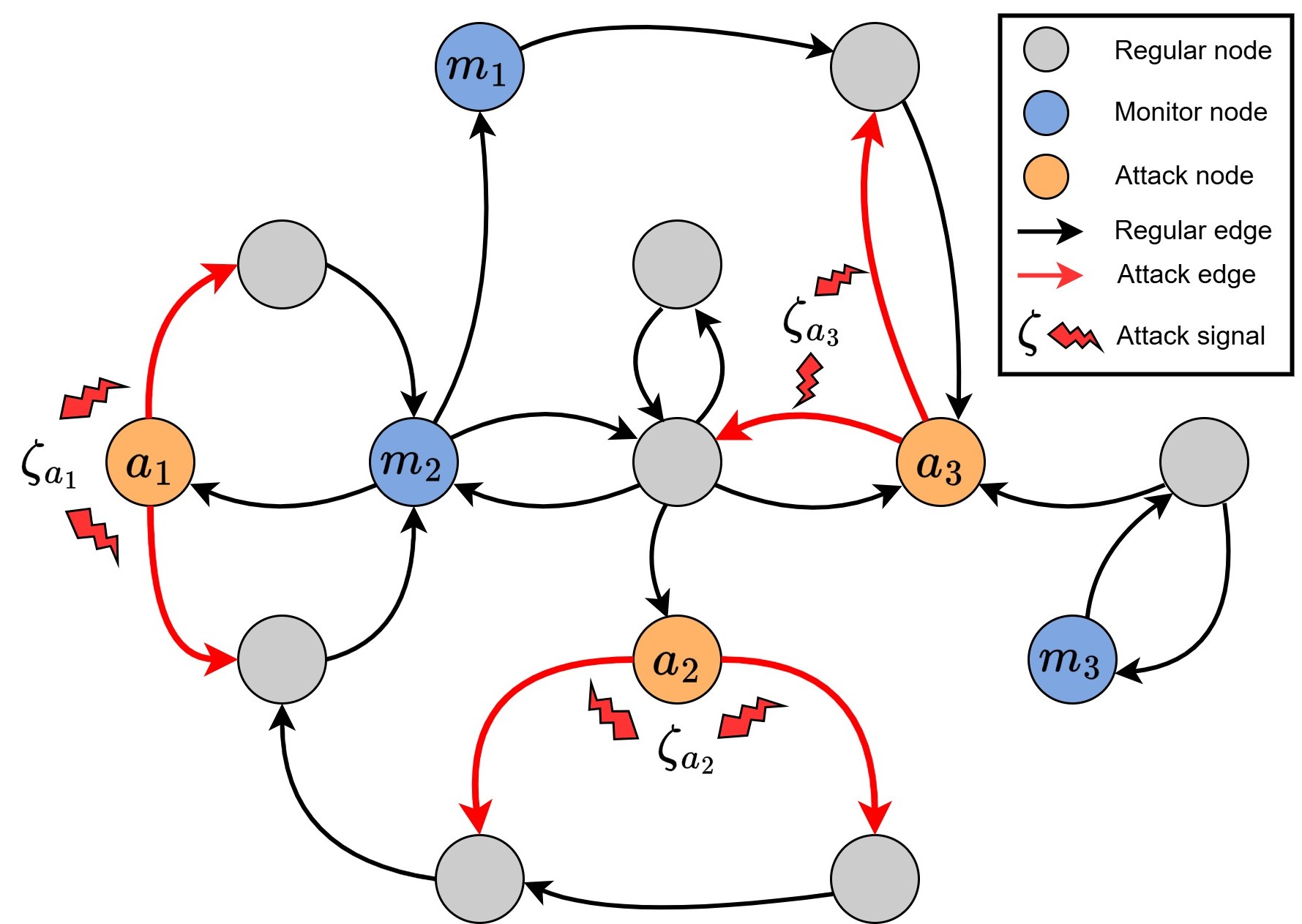}
    \caption{A networked control system under stealthy data injection attacks. An adversary injects attack signals into the information broadcast from orange nodes to their neighbors while a defender monitors the outputs of blue nodes.}
    \vspace{-2pt}
    \label{fig:problem}
\end{figure}
\section{Problem Formulation}
\label{sec:problem}
In this section, we first describe a positive networked control system with a network performance loss under stealthy attacks. 
Secondly, the adversary and the defender are modeled with opposing objectives.
Finally, we formulate the three main research problems that will be studied in this paper. 
\subsection{Positive networks under attacks}
We first introduce a positive networked control system in a normal operation without attacks and describe the resources of the adversary and the defender. Subsequently, the system modeling under attacks is formulated.  
\subsubsection{Positive networks without attacks}
Consider a strongly connected digraph $\Gc \triangleq (\Vc, \, \Ec, \, A, \, \Theta)$ with $N$ nodes and positive weights, the state-space model of a one-dimensional node $i$ is described:
\begin{align}
	\dot x_i(t) &= -\theta_i x_i(t) + \sum_{j \in \Nc_i} A_{ij} \big(x_j(t) - x_i(t)\big), \non \\
    &\hspace{3cm}
 ~ i \in \Vc = \bigl\{1,\,2,\ldots,\,N\bigr\},
	\label{sys:xi}
\end{align}
where $x_i(t) \in \Rbb_{\geq 0}$ is the state of node $i$.
The second term of \eqref{sys:xi} represents the information broadcast to node $i$ from its neighbors over a communication network, which will be discussed further when the network is under attack.
For convenience, let us denote $x(t) \triangleq \big[x_1(t),~x_2(t),\ldots,~x_N(t)\big]^\top$. Therefore, the model of the network \eqref{sys:xi} without attacks can be rewritten as follows:
\begin{align}
    \dot x(t) = -L x(t), \label{sys:x_healthy}
\end{align}
where $L$ is the in-degree Laplacian matrix representing the graph $\Gc$.

Akin to robust control, the network performance loss of the entire network for a given, possibly infinite, time horizon $[0,H]$ is formulated as follows:
\begin{align}
    J \triangleq \norm{p}_{\Lc_2[0,H]}^2.
    \label{sys:J}
\end{align}
Here, the performance loss $p(t) \in \Rbb^N$ is defined as follows:
\begin{align}
    p(t) = W x(t), \label{sys:pi}
\end{align}
where $W \triangleq \textbf{diag}([w_i])$ and $w_i \in \Rbb_{> 0}$ is a given weighting factor, discussed more in Section~\ref{sec:robust_security}.
\subsubsection{Defense resources}
To detect potential malicious activities, the defender is allowed to select at most $\beta$ nodes from the node set $\Vc$ to form a set of monitor nodes, denoted as
\begin{align}
    \Mc \triangleq \{m_1,m_2,\ldots,m_{| \Mc |} \} \subset \Vc, ~ |\Mc| \le \beta. \label{monitor_set}
\end{align}
More specifically, the defender monitors the following output measurements:
\begin{align}
    y_m(t) = e^\top_m x(t),~~ \forall \, m \in \Mc. \label{sys:ym}
\end{align}
At each monitor node $m \in \Mc$, a corresponding alarm threshold $\delta_m \in \Rbb_{>0}$ is assigned. The defender detects the presence of the adversary if the output energy for a given time horizon $[0,H]$ of at least one monitor node crosses its corresponding alarm threshold, i.e.,
\begin{align}
    \norm{y_m}^2_{\Lc_2[0,H]} > \delta_m^2. \label{alarm_deltam}
\end{align}
For later use, let us denote the vector of all the alarm thresholds as $\delta \triangleq [\delta_1,\,\delta_2,\ldots,\,\delta_N]^\top$. 
\subsubsection{Adversary resources}
In critical conditions, the adversary gains access to the network and acquires enough information about the system parameters and defense strategies. Then, the adversary conducts the following attack strategies.

The adversary selects exactly $\alpha~(\alpha \leq N)$ nodes on which to conduct false data injection attacks \textit{on the information broadcast from these $\alpha$ attack nodes to their neighbors} (the orange nodes in Fig.~\ref{fig:problem}). More specifically, these $\alpha$ nodes are not directly affected by attacks, but their neighboring nodes are. Henceforth, these $\alpha$ nodes are called attack nodes, which are defined in the following.

Let us denote a set of $\alpha$ attack nodes as follows:
\begin{align}
    \Ac \triangleq \{ a_1, a_2, \ldots, a_{\alpha} \} \subset \Vc. \label{attack_set}
\end{align}
For each attack node $a_i \in \Ac$, the adversary designs an additive non-negative attack signal $\zeta_{a_i}(t) \in \Rbb_{\geq 0}$ into the information broadcast from the attack node $a_i$ to all its out-neighbors, which is assumed to belong to $\Lc_{2}$. This assumption can be translated into the following bounded energy constraint:
\begin{align}
   \epsilon \norm{\zeta_{a_i}}^2_{\Lc_2[0,\infty]}  \leq 1, ~\forall \, a_i \in \Ac, \label{zeta_bounded}
\end{align}
where the parameter $\epsilon \in \Rbb_{\geq 0}$ is used by the defender to conduct an analysis of the attack impact and to design monitoring policy later in the paper. Note that, for considering more powerful adversaries, the defender can remove this energy assumption \eqref{zeta_bounded} by setting $\epsilon = 0$, which is further discussed in Section~\ref{sec:graph_security}.
Consequently, the state dynamics in \eqref{sys:xi} under false data injection attacks can be represented as follows:
\begin{align}
    \dot x^a_j(t) \triangleq \dot x_j(t) + \sum_{a_i \in  \Nc_{j} \cap \Ac} A_{j \, a_i} \zeta_{a_i}(t), \label{sys:uia}
\end{align}
where $A_{j \, a_i}$ is the $(j, a_i)$-element of the adjacency matrix $A$ and the superscript ``$a$'' stands for signals possibly subjected to attacks.

Given the attack set $\Ac$ in \eqref{attack_set}, let us denote the attack input matrix $E_\Ac = [e_{a_1},e_{a_2},\ldots,e_{a_{\alpha}}]$ and the attack signal vector as $\zeta(t) \triangleq [\zeta_{a_1}(t),\zeta_{a_2}(t),\ldots,\zeta_{a_{\alpha}}(t)]^\top \in \Rbb_{\geq 0}^{\alpha}$.


\subsubsection{Positive networks under attacks}
Given the above descriptions of the healthy network \eqref{sys:x_healthy}, the healthy network performance loss \eqref{sys:J}-\eqref{sys:pi}, the monitor outputs \eqref{sys:ym}, the adversary strategy \eqref{sys:uia}, and a set of attack nodes $\Ac$ in \eqref{attack_set} which corresponds to an attack input matrix $E_\Ac$, the positive network under false data injection attacks can be described as follows:
\begin{align}
    \dot x^a(t) &= - L x^a(t)  + A E_\Ac  \zeta(t), \label{sys:xa} \\
    p^a(t) &= W x^a(t), \label{sys:pa} \\
    y^a_m(t) &= e_m^\top x^a(t), ~ \forall m \in \Mc, \label{sys:yma}
\end{align}
As a consequence, the malicious impact of attacks on the network performance loss is formulated as follows:
\begin{align}
    J^a \triangleq \norm{p^a}_{\Lc_2[0,H]}^2.
    \label{sys:Ja}
\end{align}
It is worth noting that the network performance loss under attacks \eqref{sys:Ja} is computed based on the signals under attacks \eqref{sys:pa} and is different from the network performance loss \eqref{sys:J} without attacks. In the sequel, we mainly consider \eqref{sys:Ja}. 

Given that the graph $\Gc$ is strongly connected and self-loop control gains in \eqref{sys:xi} exist, $-L$ in \eqref{sys:xa} is Hurwitz. As a result, the system can be assumed to converge to its equilibrium before being exposed to attacks. 
Let us make use of the following assumptions.
\begin{Assumption}
    \label{assumption:x0}
    The system \eqref{sys:xa} is at its equilibrium $x_e = 0$ before being affected by the attack signal $\zeta(t)$, i.e., $x^a(0) = x_e = 0$.  \QET
\end{Assumption}
\begin{Assumption}
    \label{assumption:attack_node_independence}
    The matrix $A E_\Ac$ in \eqref{sys:xa} contains linearly independent columns. \QET
\end{Assumption}

Assumption~\ref{assumption:attack_node_independence} enables us to consider attack signals that have linearly independent effects on the network. A linearly dependent column in $A E_\Ac$ and its corresponding attack signal can be eliminated since its effect on the network can be replaced with an extra linear combination of the other attack signals without loss of generality. In the following, we describe the models of the adversary and the defender, where they have opposing objectives. 


\subsection{Adversary model}
The purpose of the adversary is to maximally degrade the network performance loss \eqref{sys:Ja} subjected to the network model \eqref{sys:xa}-\eqref{sys:yma} while remaining stealthy to the defender.
This adversarial purpose allows us to mainly focus on stealthy injection attacks, which will be defined in the following. 
\begin{Definition}
    [Stealthy injection attacks] \label{def:stealthy_attacks}
    Consider the structure of the continuous LTI system
    \eqref{sys:xa}-\eqref{sys:yma} with monitor outputs $y_m^a(t) = e_m^\top x^a(t)$ for every $m \in \Mc$, which is a set of monitor nodes. The attack $\zeta(t)$ on the system \eqref{sys:xa}-\eqref{sys:yma} is called a stealthy attack if the following condition $\norm{y_m^a}^2_{\Lc_2} \leq \delta_m^2$ holds for all $m \in \Mc$. \QET
\end{Definition}

In Fig.~\ref{fig:example_stealthy}, we show an example that distinguishes a stealthy attack from a detected attack. Due to the fact that a detected attack can be mitigated responsively, stealthy attacks are of interest to us afterward (see the discussion on the importance of stealthiness in \cite[Section II.E]{umsonst2021bayesian}).

Based on Definition~\ref{def:stealthy_attacks}, the worst-case disruption of stealthy injection attacks on the network performance loss \eqref{sys:Ja}, given a monitor set $\Mc$ and an attack set $\Ac$, is formulated as follows:
\begin{align}
    Q(\Mc,\Ac) \triangleq 
    ~\sup_{\zeta \in \Rbb_{\geq 0}^{\alpha}, \zeta \in \Lc_{2e}}&~
    \norm{p^a}_{\Lc_2}^2 
    \label{Q_sup} \\
    \text{s.t.}~~~~~&~
    \eqref{sys:xa}-\eqref{sys:yma}, \, x^a(0) = 0, \non \\
    &~ \norm{y_m^a}^2_{\Lc_2} \leq \delta_m^2,~ \forall \, m \in \Mc, \non \\
    &~~ \epsilon \norm{\zeta_{a_i}}^2_{\Lc_2[0,\infty]}  \leq 1, ~\forall \, a_i \in \Ac. \non
\end{align}

The worst-case disruption \eqref{Q_sup} is also called the Output-to-Output gain security metric proposed in \cite{teixeira2015strategic}.
The non-convex optimization problem \eqref{Q_sup} can be computed by leveraging its dual form and the dissipative systems theory \cite{moylan2014dissipative}, as the same procedure reported in \cite{teixeira2015strategic}. 
However, such a procedure does not scale well and is therefore unsuitable for very large networks (see the discussion on the time complexity \cite[Section IV]{nguyen2025scalable}). By leveraging the positivity of the network considered in this paper, we aim for an alternative way to quantify \eqref{Q_sup} with a lighter computing resource requirement, which is expected to suit large networks. 
Therefore, let us present the following problem that will be studied.
\begin{Problem}
    [Scalable assessment] \label{problem:scalable}
    Given a positive network under attack, depicted in Fig.~\ref{fig:problem}, representing the system \eqref{sys:xa}, a monitor set $\Mc$ in \eqref{monitor_set}, an attack set $\Ac$ in \eqref{attack_set}, and the worst-case disruption \eqref{Q_sup}, provide a scalable computation for \eqref{Q_sup}, which is largely independent of the network size.
    \QET
\end{Problem}

\begin{figure}[!t]
    \centering
    \includegraphics[width=0.95\linewidth]{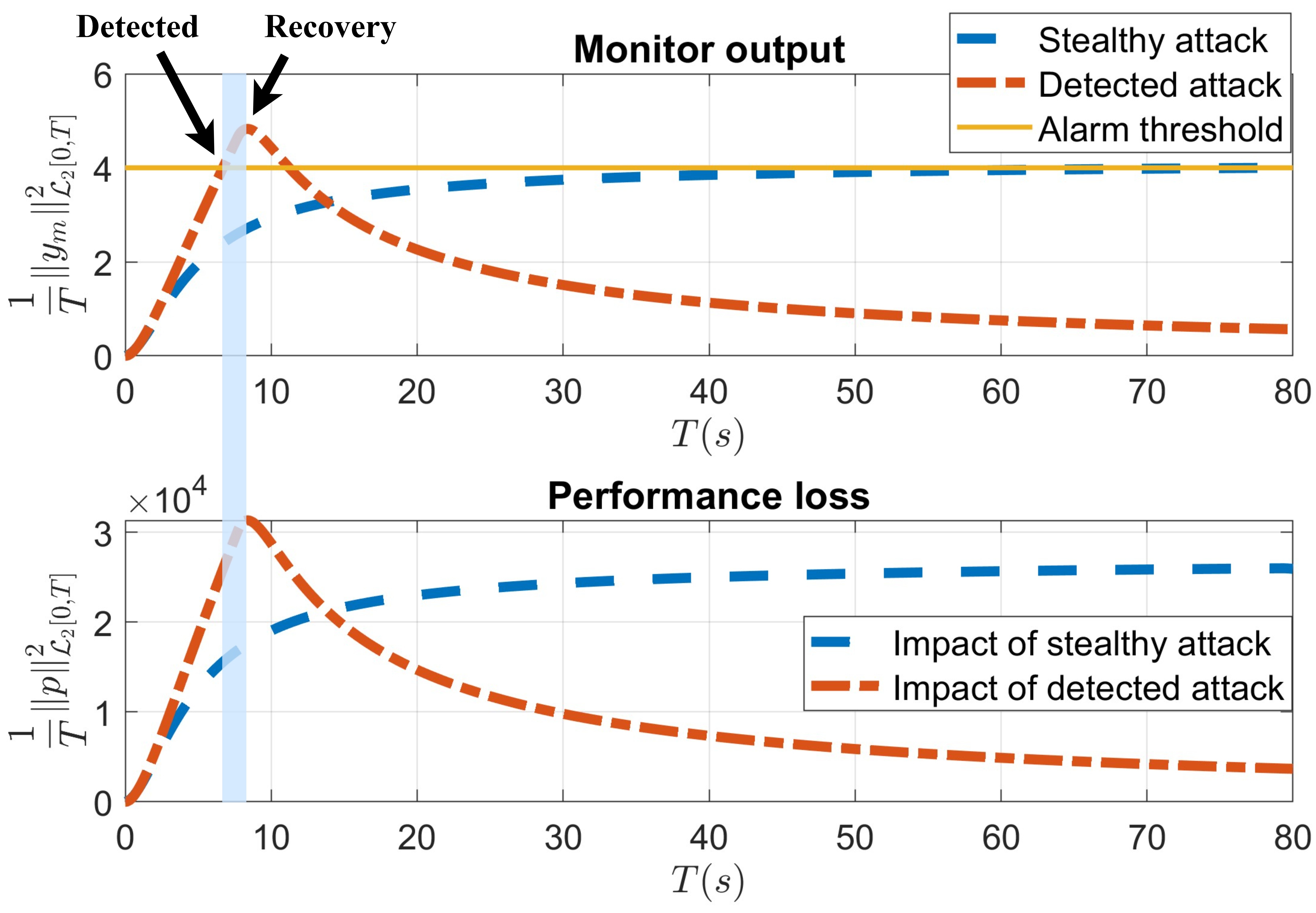}
    \caption{Illustration of stealthy (blue dashed) and detected (red dashed-dotted) attacks relative to the alarm threshold (yellow solid) in the monitor output (top) and their corresponding performance loss (bottom). For illustration purposes, we assume that the system has a recovery mechanism after an attack is detected, i.e., setting the attack signal to zero onward, illustrated by a light-blue zone. The detected attack might cause higher initial performance loss before being detected, but its impact diminishes after being detected. Meanwhile, the stealthy attack remains below the alarm threshold and causes a steady disruption in the system.}
    \label{fig:example_stealthy}
\end{figure}

\subsection{Defender model}
In practice, the defender should design strategic defense plans against possible malicious activities due to the fact that the defender is normally short of information when adversaries attack the system \cite{falliere2011w32, Havex2014, kshetri2017hacking, Israel2020, miller2021looking}. For a fixed attack set $\Ac$, the defender chooses a monitor set $\Mc$ such that it minimizes the worst-case disruption \eqref{Q_sup}, i.e., 
\begin{align}
    \min_{\Mc}& ~ Q(\Mc, \Ac) \label{def_min_Mc} \\
    \text{s.t.}& ~ |\Mc| \le \beta, \non
\end{align}
where $Q(\Mc, \Ac)$ is defined in \eqref{Q_sup}. Furthermore, to have a robust defense policy, the defender considers the following optimization problem:
\begin{align}
    \min_{\Mc, Q_0 > 0}& ~ Q_0 \label{def_min_Mc_robust} \\
    \text{s.t.}~~~& ~ |\Mc| \le \beta, ~Q(\Mc, \Ac_i) \leq Q_0~~\forall \Ac_i, |\Ac_i| = \alpha. \non 
\end{align}
Here, we assume that the defender can only choose at most $\beta$ monitor nodes. Otherwise, the solutions to \eqref{def_min_Mc}-\eqref{def_min_Mc_robust} are trivial by selecting all the nodes.
We aim to solve \eqref{def_min_Mc}-\eqref{def_min_Mc_robust} by addressing the following two problems via a graph-theoretic perspective.

\begin{Problem}
    [Graph-theoretic optimal selection]
    \label{problem:opt_sel}  
    Given a digraph $\Gc$ describing the system under attacks \eqref{sys:xa}-\eqref{sys:yma}, an attack set $\Ac$, and the worst-case disruption \eqref{Q_sup}, provide a graph-theoretic selection of a monitor set $\Mc$ which is the optimal solution to \eqref{def_min_Mc}.
    \QET
\end{Problem}
\begin{Problem}
    [Graph-theoretic robust optimal selection]
    \label{problem:opt_sel_robust}  
    Given a digraph $\Gc$ describing the system under attacks \eqref{sys:xa}-\eqref{sys:yma}, and the worst-case disruption \eqref{Q_sup}, provide a graph-theoretic selection of a monitor set $\Mc$ which is the optimal solution to \eqref{def_min_Mc_robust}.
    \QET
\end{Problem}

Problem~\ref{problem:opt_sel} concerns monitor allocation for a fixed attack scenario, whereas Problem~\ref{problem:opt_sel_robust} seeks a robust monitor configuration against all admissible attack scenarios. Clearly, Problem~\ref{problem:opt_sel} is a particular version of Problem~\ref{problem:opt_sel_robust} when $\Ac_i$ is the only one admissible attack set. Later in the paper, the simplicity of Problem~\ref{problem:opt_sel} enables us to study it with the aim of laying a foundation for addressing Problem~\ref{problem:opt_sel_robust}. While Problem~\ref{problem:opt_sel} admits a more explicit graph-theoretic characterization,
Problem~\ref{problem:opt_sel_robust} will be addressed through a heuristic search guided by the same structural insights.

\begin{Remark}
\label{rem:prob_clarify}
    The results developed for Problems~\ref{problem:scalable}-\ref{problem:opt_sel_robust} enable us not only to quantify the output-to-output gain security metric in a more efficient manner but also to show a graph-theoretic interpretation of good/bad monitor nodes and high/low likely attack nodes. This interpretation can assist us in dealing with allocating monitor nodes \cite{nguyen2025scalable, nguyen2024security,li2018false,yuan2019stackelberg,umsonst2021bayesian} and protecting unmonitored nodes \cite{anand2023risk,shukla2022robust} in large networks. It is worth noting that the optimization problems \eqref{def_min_Mc}-\eqref{def_min_Mc_robust} have been solved by using a mixed-integer SDP in \cite[Theorem 1]{nguyen2025scalable}. Due to the fact that a mixed-integer SDP is NP-hard in general and scales badly with network sizes, we aim at dealing with Problems~\ref{problem:opt_sel}-\ref{problem:opt_sel_robust} without directly solving such an optimization problem, better suited to large networks. \QET
\end{Remark}

To facilitate addressing Problems~\ref{problem:opt_sel}-\ref{problem:opt_sel_robust}, we introduce a graph-theoretic property, which will be regularly employed throughout the remainder of the paper, in the following subsection.

\subsection{Katz-like centrality measures}
Among popular network centrality measures \cite{freeman1977set} such as degree, betweenness, and closeness, they focus on graph topologies of the network rather than state dynamics evolution over the network. Consequently, these centrality measures are not directly applicable to the security context, which involves graph topologies, system dynamics, and security parameters, as reported in our previous work \cite{nguyen2024centrality} and existing work \cite{riehl2017centrality,forsberg2023power}.
Instead, the Katz centrality measure \cite{katz1953centrality} offers a better balance between the graph topology and the state dynamics evolution, making it a promising tool for security problems that involve both network structure and system dynamics. More specifically, this measure indicates how much the influence of one node has on the other nodes in the graph, which better suits the security problems we study.
Let us introduce two adapted versions of the Katz centrality measure in the following:
\begin{Definition}
    [Katz-like centrality measure] \label{def:Katz}
    Given a digraph $\Gc$ with an in-degree adjacency matrix $A$ and an in-degree matrix $D_{\rm in}$, the monitor Katz-like centrality matrix $K_\delta$ and the performance Katz-like centrality matrix $K_W$, respectively, are defined as follows:
    \begin{align}
        K_\delta &\triangleq \textbf{diag}(\delta)^{-1} \sum_{i = 1}^\infty ~ (D_{\rm in}^{-1} A)^{i}, \label{Katz_monitor} \\
        K_W &\triangleq W \sum_{i = 1}^\infty ~ (D_{\rm in}^{-1} A)^{i}, \label{Katz_performance}
    \end{align}
    where $\delta \in \Rbb_{>0}^N$ is the alarm thresholds of all the nodes and $W$ is the performance matrix of all the nodes. 
    \QET
\end{Definition}

The centrality matrices in \eqref{Katz_monitor}-\eqref{Katz_performance} are adapted by adding system and security parameters, i.e., $W$ and $\delta$, to the original Katz centrality matrix in \cite{katz1953centrality}. This adaptation plays an important role in the security problem considered since the original Katz centrality matrix was proposed for graphs in general, without considering dynamical systems and security problems. These two adapted versions of the Katz centrality measure \eqref{Katz_monitor}-\eqref{Katz_performance} will be employed in addressing the worst-case disruption \eqref{Q_sup} in the following section and to guide the monitor-allocation analysis in Section~\ref{sec:graph_security}.


\section{Katz Centrality-based Security Assessment}
\label{sec:robust_security}
This section studies Problem~\ref{problem:scalable}. In the first part, we show that the worst-case disruption \eqref{Q_sup} is upper-bounded by a tractable SDP problem with the help of the dissipative systems theory \cite{moylan2014dissipative}. The second part of the section provides a sufficient condition under which \eqref{Q_sup} equals its resulting upper bound and is exactly computed in a more efficient manner, a solution to Problem~\ref{problem:scalable}.

\subsection{The worst-case disruption upper bound}
Followed by our previous results in \cite[Lemma 3]{nguyen2025scalable}, the worst-case disruption \eqref{Q_sup} is upper bounded by a solution to an SDP problem in the following lemma. 
\begin{Lemma}
    [Worst-case disruption upper bound]
    \label{lem:Q_computation}
    Given the network under attack \eqref{sys:xa}-\eqref{sys:yma}, the monitor set $\Mc$, and the attack set $\Ac$, the worst-case disruption \eqref{Q_sup} is upper bounded by a solution to the following SDP problem:
    \begin{align}
        Q(\Mc, \Ac) \leq \Gamma(P \in \Sbb^N_{\succeq 0}), \label{Q_upper_bound}
    \end{align}
    where
    \begin{align}
    \hspace{-0.1cm}
        \Gamma(P \in \Sbb^N_{\succeq 0}) = &
    \min_{\gamma_m, \psi, P}~ \sum_{m \in \Mc} \gamma_m \delta_m^2  + \textbf{1}^\top \psi \label{Q_min_sdp} \\ 
    \text{s.t.}~~~& \gamma_m  \in  \Rbb_{\geq 0}, \psi \in \Rbb^\alpha_{\geq 0}, P \in \Sbb^N_{\succeq 0}, \non \\
    &
    \ba{cc}
    -L^\top P - P L + W^2 & P A E_\Ac \\
    E_\Ac^\top A^\top P & - \epsilon \, \text{diag}(\psi)
    \ea \non \\
    &
    - \sum_{m \in \Mc} \gamma_m \textbf{diag} \Bigg(\ba{c} 
    e_m  \\ 0 \ea \Bigg) \preceq 0, \non
    \end{align}
    where $e_m$ is a canonical basis vector representing the monitor node $m \in \Mc$ and $\textbf{1}$ is an all-ones vector.
\QET
\end{Lemma}
\begin{proof}
    Given that $Q(\Mc, \Ac)$ mainly considers positive signals,
    we denote $\tilde Q(\Mc, \Ac)$ as the worst-case disruption when considering all possible attack signals. Consequently, since $\tilde Q(\Mc, \Ac)$ has a larger feasibility set than $Q(\Mc, \Ac)$ with maximizing the same objective function, one has
    \begin{align}
        Q(\Mc, \Ac) \leq \tilde Q(\Mc, \Ac). 
    \end{align}
    By employing the dissipative systems theory \cite{moylan2014dissipative}, $\tilde Q(\Mc, \Ac)$ is upper bounded by the solution to the SDP problem $\Gamma(P \in \Sbb^N_{\succeq 0})$ in \eqref{Q_min_sdp} as reported in our previous work \cite{nguyen2024security} with a storage function $V = (x^a(t))^\top P x^a(t)$ and a supply rate $s(t) = \sum_{m \in \Mc} \gamma_m \norm{y_m^a(t)}_2^2 + \sum_{i = 1}^\alpha e_i^\top \psi \, \epsilon \norm{e_i^\top \zeta(t)}_2^2 - \norm{p^a(t)}_2^2$, where $P \in \Sbb^{N}_{\succeq 0}$. This completes the proof.
\end{proof}

In the SDP problem \eqref{Q_min_sdp}, an $(N \times N)$-symmetric variable $P$ is used where $N$ is the network size. This matrix $P$ contains $(N^2+N)/{2}$ decision variables. Clearly, the number of variables in \eqref{Q_min_sdp} significantly increases as the network size increases, making it impractical for large networks.
Regarding computational complexity, solving \eqref{Q_min_sdp} with generic
interior-point methods typically scales poorly with the network size, with a
worst-case complexity on the order of $O(N^6)$ (see \cite[Chapter~11.8.3]{boyd2004convex}).
Moreover, the SDP \eqref{Q_min_sdp} only serves as an upper bound for \eqref{Q_sup}, which might not be directly applicable for security assessment and security deployment.
This computational drawback motivates the study of Problem~\ref{problem:scalable} in the next subsection.
In contrast to the computational complexity of the SDP problem \eqref{Q_min_sdp},
the following subsection shows a sufficient condition under which the worst-case disruption \eqref{Q_sup} can be computed in a more efficient way.

\subsection{Exact computation for the worst-case disruption}

Let us make use of the following robustness measure:
\begin{align}
    Q_{\infty}(\Ac) \triangleq Q(\emptyset,\Ac), \label{Q_sup_im_Hinf}
\end{align}
which stands for the worst-case disruption against no monitor nodes. From a robust control perspective, the notation $Q_{\infty}(\Ac)$ measures the worst-case impact of attacks on the system.
It is worth noting that \eqref{Q_sup_im_Hinf} differs from the classical $\Hc_\infty$ norm since it considers the element-wise bounded energy for input signals. This metric is also called the diagonal norm-bound in \cite[Chapter 6]{boyd1994linear}.
Recall the worst-case disruption \eqref{Q_sup},
the following lemma shows its upper bound in terms of the robustness measure $Q_{\infty}(\Ac)$ in \eqref{Q_sup_im_Hinf}.

\begin{Lemma}
    [Sufficient robustness] 
    \label{lem:WboundQ} 
    Consider \eqref{Q_sup} and \eqref{Q_sup_im_Hinf}. Suppose that the following condition holds
    \begin{align}
        Q_\infty(\Ac) \leq \min_{i \in \Vc} w_{i}^2 ~\min_{m \in \Mc} \delta_m^2,
        \label{cond:W2_Hinf}
    \end{align} 
where $Q_{\infty}(\Ac)$ is defined in \eqref{Q_sup_im_Hinf}, $w_{i}$ is the $i$-th diagonal element of the performance weighting matrix $W$, and $\delta_m$ is the alarm threshold at monitor node $m$ of the monitor set $\Mc$. Then, the following inequality holds true
\begin{align}
    Q(\Mc,\Ac) I \preceq \min_{m \in \Mc} \delta_m^2 W^2. \label{cond:in_KYP}
\end{align} 
\end{Lemma}

\begin{proof}
    By construction, the optimal solution to \eqref{Q_sup_im_Hinf} is not smaller than that to \eqref{Q_sup}, i.e., $Q_{\infty}(\Ac) = Q(\emptyset, \Ac) \geq Q(\Mc,\Ac)$ for all monitor set $\Mc$. On the other hand, if the condition \eqref{cond:W2_Hinf} holds, one has $Q_\infty(\Ac) I \preceq \min_{m \in \Mc} \delta_m^2 W^2$. As a result, one obtains $Q(\Mc,\Ac) I \preceq \min_{m \in \Mc} \delta_m^2 W^2$ for any monitor set $\Mc$, concluding the proof.
\end{proof}



Lemma~\ref{lem:WboundQ} shows that a sufficient robustness of the network in the absence of monitor nodes essentially yields an upper bound for the worst-case disruption in the presence of monitor nodes in terms of system parameters.
This upper bound further enables us to employ a result on the celebrated Kalman-Yakubovich-Popov Lemma for positive systems in \cite[Theorem 1]{rantzer2015kalman} with the two purposes, which are 1) addressing Problem~\ref{problem:scalable}; and 2) establishing the connection between the worst-case impact \eqref{Q_sup} and the Katz-like centrality measures \eqref{Katz_monitor}-\eqref{Katz_performance} in the following theorem.

\begin{Theorem} [Katz centrality-based security assessment] \label{thm:robust_security}
    Consider the relationship \eqref{Q_upper_bound} and suppose 
    the condition \eqref{cond:W2_Hinf} holds. Then, the inequality in \eqref{Q_upper_bound} becomes equality, and the SDP problem on the right-hand side can be exactly solved with a non-negative diagonal $P$, i.e.,
    \begin{align}
        Q(\Mc, \Ac) = \Gamma(P \in \Sbb^N_{\succeq 0}) =  \Gamma(P \succeq 0, P ~\text{is diagonal}). \label{Q_sdp_P_diagonal}
    \end{align}
    Further, $\Gamma(P \succeq 0, P ~\text{is diagonal})$ is even equivalent to the following SDP problem   
    \begin{align}
       \Gamma(P \succeq 0, P ~\text{is diagonal}) = \Gamma(\text{non-}P), \label{Q_sdp_Ainverse}
     \end{align}
     where
     \begin{subequations} \label{Q_min_sdp_short}
     \begin{align} 
    &  \Gamma(\text{non-}P) =  
    \min_{\gamma_m \in \Rbb_{\geq 0}, \psi \in \Rbb_{\geq 0}^\alpha}  \sum_{m \in \Mc} \gamma_m \delta_m^2 +  \textbf{1}^\top \psi \label{Q_min_sdp_short_obj} \\ 
     &\text{s.t.}~~ 
     ( K_W E_\Ac)^\top K_W E_\Ac - \epsilon \, \text{diag}(\psi) \preceq 
     \non \\
     & \hspace{1.8cm} 
     \sum_{m \in \Mc} \gamma_m \delta_m^2 (E_\Ac^\top K_\delta^\top e_m) (E_\Ac^\top K_\delta^\top e_m)^\top
     \label{Q_min_sdp_short_cons}
    \end{align}
    \end{subequations}
    with 
    $K_\delta$ defined in \eqref{Katz_monitor}, and $K_W$ defined in \eqref{Katz_performance}. 
    \QET
\end{Theorem}

\begin{proof}
    See Appendix~\ref{app:thm:robust_security_pf}.
\end{proof}

By observing \eqref{Q_min_sdp_short} in Theorem~\ref{thm:robust_security}, the robustness of the network represented in \eqref{cond:W2_Hinf} assists the computation of \eqref{Q_min_sdp} in dropping the matrix variable $P$, which contains $(N^2+N)/2$ decision variables. This drop significantly reduces the computational complexity. In \eqref{Q_min_sdp_short}, the decision variables consist of one non-negative scalar $\gamma_m$ per monitor node together with the $\alpha$-dimensional vector $\psi$, which is significantly smaller than the variable set in \eqref{Q_min_sdp}. 
Furthermore, the dimension of the linear matrix inequality constraint in \eqref{Q_min_sdp_short} is the same as the number of attack nodes and is independent of the size of the network. As a result, the computation \eqref{Q_min_sdp_short} does scale well with the network size, suitable for large networks.

For a further graph-theoretic observation, the left-hand side (LHS) of \eqref{Q_min_sdp_short_cons} represents the impact of attack nodes on the network performance through the performance Katz-like centrality measure $K_W$, while the right-hand side (RHS) of \eqref{Q_min_sdp_short_cons} represents the impact of attack nodes on monitor nodes through the monitor Katz-like centrality measure $K_\delta$. By leveraging this graph-theoretic interpretation, the next section studies  Problem~\ref{problem:opt_sel} and develops a heuristic search for sub-optimal monitor nodes as a proxy approach to Problem~\ref{problem:opt_sel_robust}.

\begin{Remark}
		It is worth noting that \eqref{Q_sdp_P_diagonal} and \eqref{Q_sdp_Ainverse} hold if the sufficient condition \eqref{cond:W2_Hinf} holds true. When this condition does not hold, $\Gamma(P \succeq 0, P ~\text{is diagonal})$ serves as an upper-bound for $\Gamma(P \in \Sbb^N_{\succeq 0})$ (see more details in the proof of Theorem~\ref{thm:robust_security}). Meanwhile, $\Gamma(\text{non-}P)$ serves as a lower bound for $\Gamma(P \in \Sbb^N_{\succeq 0})$ since $\Gamma(\text{non-}P)$ is equivalent to only considering attack signals at zero frequency (see Proposition~\ref{thm_anders} for more details). Therefore, one always obtains the following relationship:
		\begin{align}
				\Gamma(\text{non-}P) \leq \Gamma(P \in \Sbb^N_{\succeq 0}) \leq	\Gamma(P \succeq 0, P ~\text{is diagonal}). \label{Q_min_sdp_upper_lower_bound}
	    \end{align}
	    Since $\Gamma(\text{non-}P)$ and $\Gamma(P \succeq 0, P ~\text{is diagonal})$
	    require significantly less computational resources than $\Gamma(P \in \Sbb^N_{\succeq 0})$, they are still useful in analyzing security for large networks. \QET
	\end{Remark}

\section{Katz Centrality-based Security Allocation}
\label{sec:graph_security}
In this section, we address Problems~\ref{problem:opt_sel}--\ref{problem:opt_sel_robust} by
leveraging the result presented in Theorem~\ref{thm:robust_security}.
Within Problem~\ref{problem:opt_sel}, we are interested in analyzing the existence of a solution to \eqref{Q_min_sdp_short} through the Katz-like centrality measures. This analysis supports us in providing a heuristic search for a sub-optimal monitor set, which addresses Problem~\ref{problem:opt_sel_robust}.

\subsection{Katz centrality-based geometric interpretation}

We investigate \eqref{Q_min_sdp_short} with $\epsilon$ being set to $0$ in order to study Problems~\ref{problem:opt_sel}-\ref{problem:opt_sel_robust}. This assumption enables us to relax the attack energy constraint \eqref{zeta_bounded} in \eqref{Q_sup}, with the last inequality constraint removed, to have a richer optimization problem, whose boundedness is investigated in this subsection. The LHS of \eqref{Q_min_sdp_short_cons} is considered in the following lemma.

\begin{Lemma}
    \label{lem:indepedent_influence}
    Suppose $\epsilon = 0$ and Assumption~\ref{assumption:attack_node_independence} holds true. The matrix computed in the LHS of \eqref{Q_min_sdp_short_cons} is symmetric and positive definite.
  \QET
\end{Lemma}
\begin{proof}
    Since the matrices $W$ and $L$ are non-singular, the alternative computation of $K_W$ in \eqref{Katz_performance_mod} implies that the matrix $(K_W E_\Ac)^\top K_W E_\Ac$ is symmetric positive definite when $A E_\Ac$ is of full column rank, which is obtained by Assumption~\ref{assumption:attack_node_independence}.
\end{proof}

In the following, we leverage the result of Lemma~\ref{lem:indepedent_influence} to show the geometric representation of the constraint in \eqref{Q_min_sdp_short}.
Intuitively, Lemma~\ref{lem:indepedent_influence} implies that the matrix $(K_W E_\Ac)^\top K_W E_\Ac$ is positive definite when the influence of an attack signal on the network cannot be replaced with a linear combination of the other attack signals. Clearly, this scenario is beneficial to the adversary, enabling us to consider the worst-case attack scenario afterward. From Lemma~\ref{lem:indepedent_influence}, one has
\begin{align}
    (K_W E_\Ac)^\top K_W E_\Ac = \sum_{k = 1}^\alpha \lambda_k v_k v_k^\top,~\norm{v_k}_2 = 1 ~\forall k, \label{def_EW_decomposition}
\end{align}
where $(\lambda_k, v_k)$ is an eigenpair of the matrix $(K_W E_\Ac)^\top K_W E_\Ac$. Since the matrix $(K_W E_\Ac)^\top K_W E_\Ac$ is symmetric positive definite, its eigenvectors are orthogonal, i.e., $v_k^\top v_j = 0,~ k \neq j$. On the other hand, the LHS of \eqref{Q_min_sdp_short_cons} can be used to represent the following ellipsoid:
\begin{align}
    \Ec_{W} = \{ \vartheta \in \Rbb^\alpha \,|\,
    \vartheta^\top \big( (K_W E_\Ac)^\top K_W E_\Ac  \big)^{-1} \vartheta \leq 1 \}, \label{ellipsoid_EW}
\end{align}
which has principal axes represented by $\sqrt{\lambda_k} v_k~\forall k \in \{1,2,\ldots,\alpha\}$. Therefore, the inequality constraint \eqref{Q_min_sdp_short_cons} holds true if, and only if, its RHS represents an ellipsoid $\Ec_\delta$ such that it contains $\Ec_{W}$ \cite[Chapter 8.4]{boyd2004convex}, where
\begin{align}
    &\hspace{-10pt}
    \Ec_{\delta} =  \{ \vartheta \in \Rbb^\alpha | \non \\
    &\hspace{-10pt}
    \vartheta^\top \bigg( \sum_{m \in \Mc} \gamma_m \delta_m^2 (E_\Ac^\top K_\delta^\top e_m) (E_\Ac^\top K_\delta^\top e_m)^\top  \bigg)^{-1} \vartheta \leq 1 \}. \label{ellipsoid_Ed}
\end{align}

\begin{figure}[!t]
    \centering
    \includegraphics[width=0.8\linewidth]{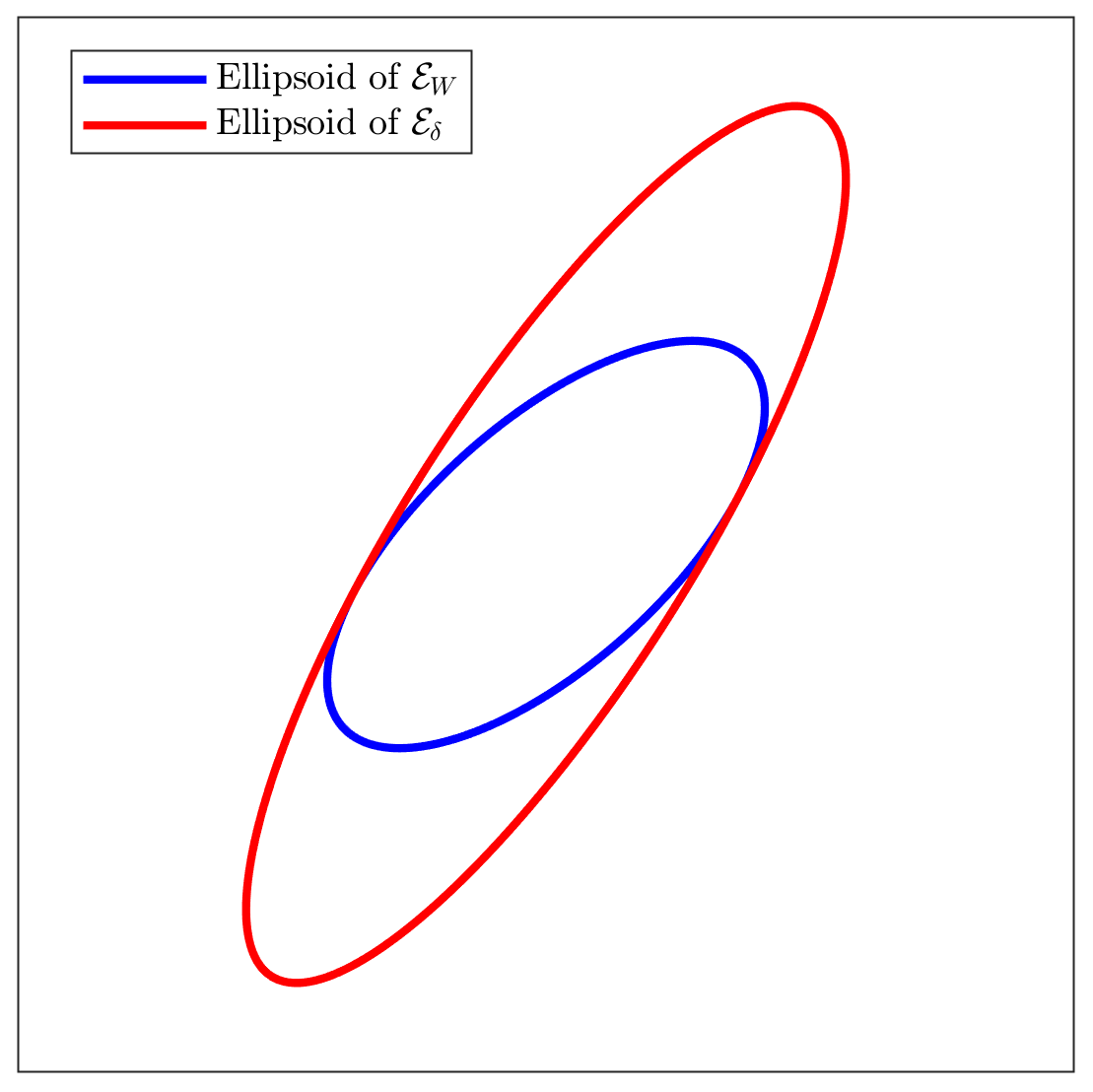}
    \caption{Illustration of the ellipsoids $\Ec_W$ and $\Ec_\delta$ with $\alpha = \beta = 2$. The constraint in \eqref{Q_min_sdp_short_ellipsoid} holds true when $\Ec_\delta$ contains $\Ec_W$.}
    \label{fig:ellipsoid_example}
\end{figure}

The intuition of the relationship between $\Ec_W$ and $\Ec_\delta$ is shown in Fig.~\ref{fig:ellipsoid_example}. 
We leverage the geometric representation of the constraint in \eqref{Q_min_sdp_short} in the previous subsection to rewrite the SDP problem \eqref{Q_min_sdp_short} into the following form:
\begin{align}
    \eqref{Q_min_sdp_short}~\text{with}~\epsilon = 0 \Lra \min_{\gamma_m \in \Rbb_{\geq 0}}& ~ \sum_{m \in \Mc} \gamma_m \delta_m^2 \label{Q_min_sdp_short_ellipsoid} \\
    \text{s.t.}~~&~~~ \Ec_W \subseteq \Ec_\delta, \eqref{ellipsoid_EW},~\eqref{ellipsoid_Ed}. \non
\end{align}
The boundedness of \eqref{Q_min_sdp_short_ellipsoid} is addressed in the following lemma.
\begin{Theorem}
    \label{lem:ellisoid_opt}
    The optimization \eqref{Q_min_sdp_short_ellipsoid} admits a solution if, and only if, there exist $\alpha$ linearly independent vectors $E_\Ac^\top K_\delta^\top e_m$ for all $m \in \Mc$. 
    \QET  
\end{Theorem}
\begin{proof}
    Based on the illustration in Fig.~\ref{fig:ellipsoid_example},
    we aim to translate the boundedness of \eqref{Q_min_sdp_short_ellipsoid} into the comparison on the number of principal axes of the ellipsoids \eqref{ellipsoid_EW} and \eqref{ellipsoid_Ed}.
    In the following, we count the number of such principal axes in (I) and (III) while their connection is made in (II).

    (I) Since $(K_W E_\Ac)^\top K_W E_\Ac$ is an $\alpha \times \alpha$ symmetric positive definite matrix according to Lemma~\ref{lem:indepedent_influence}, its ellipsoid representation $\Ec_W$ in \eqref{ellipsoid_EW} has $\alpha$ 
    principal axes.

    (II) Geometrically, the constraint of \eqref{Q_min_sdp_short_ellipsoid} holds true only if the number of principal axes of $\Ec_\delta$ is larger than or equal to the number of principal axes of $\Ec_W$, which is $\alpha$ from (I).
    This constraint does not hold true if $\Ec_\delta$ has fewer than $\alpha$ principal axes. Therefore, the necessary condition for the feasibility of \eqref{Q_min_sdp_short_ellipsoid} is that $\Ec_\delta$ has exactly $\alpha$ principal axes. 

    (III)
    Based on the definition of $\Ec_\delta$ in \eqref{ellipsoid_Ed}, this ellipsoid has principal axes represented by $E_\Ac^\top K_\delta^\top e_m$ for all $m \in \Mc$. 
    Note that two different monitor nodes might yield two parallel principal axes, resulting in a single principal axis. Consequently, the number of principal axes of $\Ec_\delta$, which is required to be $\alpha$ from (II), equals the number of the linearly independent vectors $E_\Ac^\top K_\delta^\top e_m$. The proof is concluded. 
\end{proof}


The result of Theorem~\ref{lem:ellisoid_opt} gives us a necessary and sufficient condition on the boundedness of \eqref{Q_min_sdp_short_ellipsoid}. This condition requires that the number of monitor nodes is not fewer than the number of attack nodes to guarantee the optimality of \eqref{Q_min_sdp_short_ellipsoid}. 
The failure of this condition leads to no solution to \eqref{Q_min_sdp_short_ellipsoid}, meaning that the adversary is able to conduct a stealthy attack with infinite attack impact without being detected, which can also be explained through a view from the transfer function perspective as reported in \cite{smith2015covert}.

\begin{Remark}
    [Vulnerable nodes]
Consider a fixed monitor set $\Mc$ and suppose that Assumption~\ref{assumption:attack_node_independence} holds. A direct consequence of Theorem~\ref{lem:ellisoid_opt} is that an attack set $\Ac$ yields no finite
solution to \eqref{Q_min_sdp_short} if the following condition holds true
\[
\text{rank}\!\left(E_{\Ac}^\top K_\delta^\top E_{\Mc}\right) < |\Ac|.
\]
In this case, the monitor-induced ellipsoid $\Ec_\delta$ does not have enough linearly independent growth directions to contain the attack-impact ellipsoid $\Ec_W$, and therefore the worst-case disruption becomes unbounded.
From the defender's viewpoint, this corollary highlights attack sets that are not fully covered by the deployed monitors and therefore require either additional monitor nodes or an alternative monitor placement.
\QET
\end{Remark}


The above remark highlights a practical limitation of a fixed monitor configuration, in which some admissible attack sets may remain insufficiently covered, resulting in an unbounded worst-case disruption. 
While such cases can in principle be handled by solving the exact mixed-integer SDP in~\cite{nguyen2025scalable}, this quickly becomes computationally demanding for large networks. 
Therefore, instead of directly solving Problem~\ref{problem:opt_sel_robust}, we seek a graph-theoretic heuristic that exploits Theorem~\ref{lem:ellisoid_opt} and the ellipsoidal structure of \eqref{Q_min_sdp_short_ellipsoid} to identify promising monitor nodes in the following subsection.

\subsection{Heuristic search for sub-optimal monitor nodes}
In this subsection, we are not directly solving Problems~\ref{problem:opt_sel}-\ref{problem:opt_sel_robust} since their solutions can be found by solving a mixed-integer SDP proposed in our previous work \cite{nguyen2025scalable} due to the fact that this method scales badly with the network size and the number of admissible attack sets. Instead, we aim to leverage the result of Theorem~\ref{lem:ellisoid_opt} and the geometric interpretation (see Fig.~\ref{fig:ellipsoid_example}) from the previous subsection to characterize solutions to Problems~\ref{problem:opt_sel}-\ref{problem:opt_sel_robust} with the aim of bypassing the need of solving an optimization problem. These provide us with two key observations:
\begin{enumerate}
    \item[O1)] A monitor set $\Mc$ must fulfill the necessary and sufficient condition in Theorem~\ref{lem:ellisoid_opt} for a given attack set $\Ac$ in addressing Problem~\ref{problem:opt_sel}; and for every admissible attack set in addressing Problem~\ref{problem:opt_sel_robust}.
    \item[O2)] While $\Ec_W$ is decomposed along its principal axes as in \eqref{def_EW_decomposition}, $\Ec_\delta$ is generated by a conic combination of rank-one matrices of the form $(E_\Ac^\top K_\delta^\top e_m) (E_\Ac^\top K_\delta^\top e_m)^\top$. Hence, $E_\Ac^\top K_\delta^\top e_m$ acts as a growth direction of $\Ec_\delta$ with coefficient $\gamma_m \delta_m^2$. This suggests selecting monitor nodes whose growth directions are well aligned with the principal axes of $\Ec_W$, so as to reduce the corresponding value of $\gamma_m \delta_m^2$  in the objective of \eqref{Q_min_sdp_short_ellipsoid}.
\end{enumerate}

Based on the above observations, the heuristic search is designed in the sequel. 
It is worth noting that Problem~\ref{problem:opt_sel} is a special case of Problem~\ref{problem:opt_sel_robust} when we only have one admissible attack set. For simplicity, the following explains the heuristic search for addressing Problem~\ref{problem:opt_sel}, and the search is generalized to address Problem~\ref{problem:opt_sel_robust} in Algorithm~\ref{alg:principal_monitor}.
For each admissible attack set $\Ac_i$, we compute the eigenpairs  $\{(\lambda_k(\Ac_i),v_k(\Ac_i))\}_{k=1}^{\alpha}$ of $(K_W E_{\Ac_i})^\top (K_W E_{\Ac_i})$. As discussed in Observation O2), the monitor-induced ellipsoid $\Ec_\delta$ must grow in sufficiently many independent directions to contain $\Ec_W$ according to Theorem~\ref{lem:ellisoid_opt}. As a result, we find a monitor node, denoted $m^\star_{k}(\Ac_i)$, such that it maximizes the dot product between $E^\top_{\Ac_i} K^\top_\delta e_{m^\star_{k}(\Ac_i)}$ and $v_k(\Ac_i)$ for every $v_k(\Ac_i)$, which is consistent with Observation O2). This procedure is repeated for every admissible attack set $\Ac_i$ with the aim of addressing Problem~\ref{problem:opt_sel_robust}, which motivates us to define a score vector as $s_\Mc \in \Rbb^N_{\ge 0}$. This score vector is updated when addressing each $\Ac_i$ in the following.

Consider the score vector $s_\Mc$, for each attack set 
$\Ac_i$ and each corresponding principal direction $v_k(\Ac_i)$ indexed by $k \in \{1,\dots,\alpha\}$, define
\begin{align}
	[s_\Mc]_{m_{k}^\star(\Ac_i)}
	~\leftarrow~
	[s_\Mc]_{m_{k}^\star(\Ac_i)}
	+ 
	\frac{\lambda_k(\Ac_i)}{[\Theta_{k}(\Ac_i)]_{{m_{k}^\star(\Ac_i)}}},
	\label{def_score_value}
\end{align}
where
\begin{align}
	{m_{k}^\star(\Ac_i)}
	&\triangleq 
	\argmax_{m \in \{1,2,\ldots,N\}}
	~[\Theta_{k}(\Ac_i)]_m ,
	\label{def_m_star_all}
	\\
	\Theta_{k}(\Ac_i)
	&\triangleq 
	\big[ K_\delta E_{\Ac_i} v_k(\Ac_i) \big]
	\odot
	\big[ K_\delta E_{\Ac_i} v_k(\Ac_i) \big].
	\label{def_Theta_all}
\end{align}
Here, $\odot$ stands for element-wise multiplication.
The use of the squared projection in~\eqref{def_Theta_all} follows directly from 
the quadratic structure of the SDP constraint \eqref{Q_min_sdp_short_cons}. Indeed, for any eigenpair 
$(\lambda_k(\Ac_i),v_k(\Ac_i))$ of $(K_W E_{\Ac_i})^\top (K_W E_{\Ac_i})$, 
multiplying the LMI constraint in~\eqref{Q_min_sdp_short_cons} on the left and 
right by $v_k(\Ac_i)^\top$ and $v_k(\Ac_i)$, respectively, yields
\begin{align}
	\lambda_k(\Ac_i)
	~\le~
	\sum_{m\in\Mc}
	\gamma_m \delta_m^2
	\big( v_k(\Ac_i)^\top E_{\Ac_i}^\top K_\delta^\top e_m \big)^2 .
	\label{lambda_k_projection}
\end{align}
Hence, the contribution of monitor node $m$ to covering the $k$-th principal direction $v_k(\Ac_i)$ is governed by the squared alignment term
\begin{align}
	\big( v_k(\Ac_i)^\top E_{\Ac_i}^\top K_\delta^\top e_m \big)^2,
\end{align}
which motivates the definition of $\Theta_{k}(\Ac_i)$ in 
\eqref{def_Theta_all}. Furthermore, by the definition of 
$m_{k}^\star(\Ac_i)$ in \eqref{def_m_star_all}, we have
\begin{align}
	\lambda_k(\Ac_i)
	~\le~
	\big(
	v_k(\Ac_i)^\top 
	E_{\Ac_i}^\top
	K_\delta^\top
	e_{m_{k}^\star(\Ac_i)}
	\big)^2
	\sum_{m\in\Mc} \gamma_m \delta_m^2,
\end{align}
which implies the lower bound
\begin{align}
	\frac{ \lambda_k(\Ac_i) }
	{ \big( v_k(\Ac_i)^\top 
		E_{\Ac_i}^\top 
		K_\delta^\top 
		e_{m_{k}^\star(\Ac_i)} 
		\big)^2 }
	~\le~
	\sum_{m\in\Mc} \gamma_m \delta_m^2.
	\label{lambda_k_frac}
\end{align}

Since the right-hand side of~\eqref{lambda_k_frac} is the objective of 
\eqref{Q_min_sdp_short_obj}, which is minimized in the SDP, the quantity in the left-hand side provides a directional lower bound on the objective contribution 
required to cover the $k$-th principal axis of $\Ec_W$. This justifies the score update 
in~\eqref{def_score_value}. 

By using the same computation \eqref{def_score_value} for all admissible attack sets, we have a score vector $s_\Mc$, which ranks all the nodes as candidate monitor nodes. Next, we pick $\beta$ nodes with the highest score values to form a monitor set. Since the selected monitor set must satisfy the necessary and sufficient condition in Theorem~\ref{lem:ellisoid_opt}, we need to choose more monitor nodes until the necessary and sufficient condition in Theorem~\ref{lem:ellisoid_opt} holds for every admissible attack set. This heuristic search is summarized in Algorithm~\ref{alg:principal_monitor}. Clearly, Algorithm~\ref{alg:principal_monitor} is not guaranteed to return the optimal monitor set. Instead, it yields a feasible candidate monitor set that may contain more than $\beta$ monitor nodes after the rank-repair step (see lines 14-15 in Algorithm~\ref{alg:principal_monitor}), potentially leading to a lower attack impact than a budget-constrained optimal monitor set.

\begin{algorithm}[!t]
	\caption{Principal-direction heuristic for monitor selection}
	\label{alg:principal_monitor}
	\begin{algorithmic}[1]
		\Statex{\textbf{Input:}} In-degree Laplacian $L$, weighting matrix $W$, alarm threshold $\delta$, attack budget $\alpha$, monitor budget $\beta$, and a collection of admissible attack sets $\{\Ac_i\}$.
		\Statex{\textbf{Output:}} Feasible candidate monitor set $\Mc^\star$.
		\Statex{\textbf{Initialization:}} score vector $s_\Mc = 0 \in \Rbb^N$.
		
		\State Compute $K_\delta$ from \eqref{Katz_monitor_mod}.
		\State Compute $K_W$ from \eqref{Katz_performance_mod}.
		
		\For{each attack set $\Ac_i$}
		\State Form $S_i = (K_W E_{\Ac_i})^\top (K_W E_{\Ac_i})$.
		\State Compute all eigenpairs $\{(\lambda_k(\Ac_i), v_k(\Ac_i))\}_{k=1}^{\alpha}$ of $S_i$.
		
		\For{$k = 1$ to $\alpha$}
		\State Compute:
		\[
		\Theta_k(\Ac_i ) = 
		\big[ K_\delta E_{\Ac_i} v_k(\Ac_i) \big]
		\odot
		\big[ K_\delta E_{\Ac_i} v_k(\Ac_i) \big].
		\]
		
		\State Find the best-aligned monitor node:
		\[
		m_{k}^\star(\Ac_i) \in \argmax_{m \in \{1,\dots,N\}} [\Theta_{k}(\Ac_i)]_m.
		\]
		
		\State Update the score:
		\[
		[s_\Mc]_{m_{k}^\star(\Ac_i)}
		\leftarrow
		[s_\Mc]_{m_{k}^\star(\Ac_i)}
		+ \frac{\lambda_k(\Ac_i)}{[\Theta_{k}(\Ac_i)]_{m_{k}^\star(\Ac_i)}}.
		\]
		\EndFor
		\EndFor
		
		\State Select the initial monitor set:
		$
		\Mc^\star = \{ m^{(1)},\dots,m^{(\beta)} \}
		\quad\text{indices of the $\beta$ largest entries of } s_\Mc.
		$
		\State Define $E_{\Mc^\star} = [e_{m^{(1)}}, \dots, e_{m^{(\beta)}}]$.
		
		\While{there exists $\Ac_i$ such that 
			\[
				\text{rank}(E_{\Ac_i}^\top K_\delta^\top E_{\Mc^\star}) < \alpha
			\]
			}
		\State Add to $\Mc^\star$ the highest-ranked node in $s_\Mc$ not yet selected.
		\State Update $E_{\Mc^\star}$ accordingly.
		\EndWhile
		
		\State \Return $\Mc^\star$.
	\end{algorithmic}
\end{algorithm}

\section{Simulation Results}
\label{sec:simulation}
In this section, we validate the obtained results through Erdős–Rényi random graphs. 
The simulation is performed using MATLAB R2023b with YALMIP 2021 \cite{lofberg2004yalmip}
and MOSEK solver on a personal computer with a 2.9-GHz eight-core Intel i7-10700 processor and 16 GB of RAM. 
The security parameters are chosen as $\epsilon = 0.1$, $\delta_i = 1 + \mu_i$, and $w_i = 1+\mu_i$ with random $\mu_i \in [0,\,0.1]$. In the following, we aim at validating the results presented in Theorems~\ref{thm:robust_security} and showing the performance of Algorithm~\ref{alg:principal_monitor}.
\begin{figure}[!t]
    \centering
    \includegraphics[width=\linewidth]{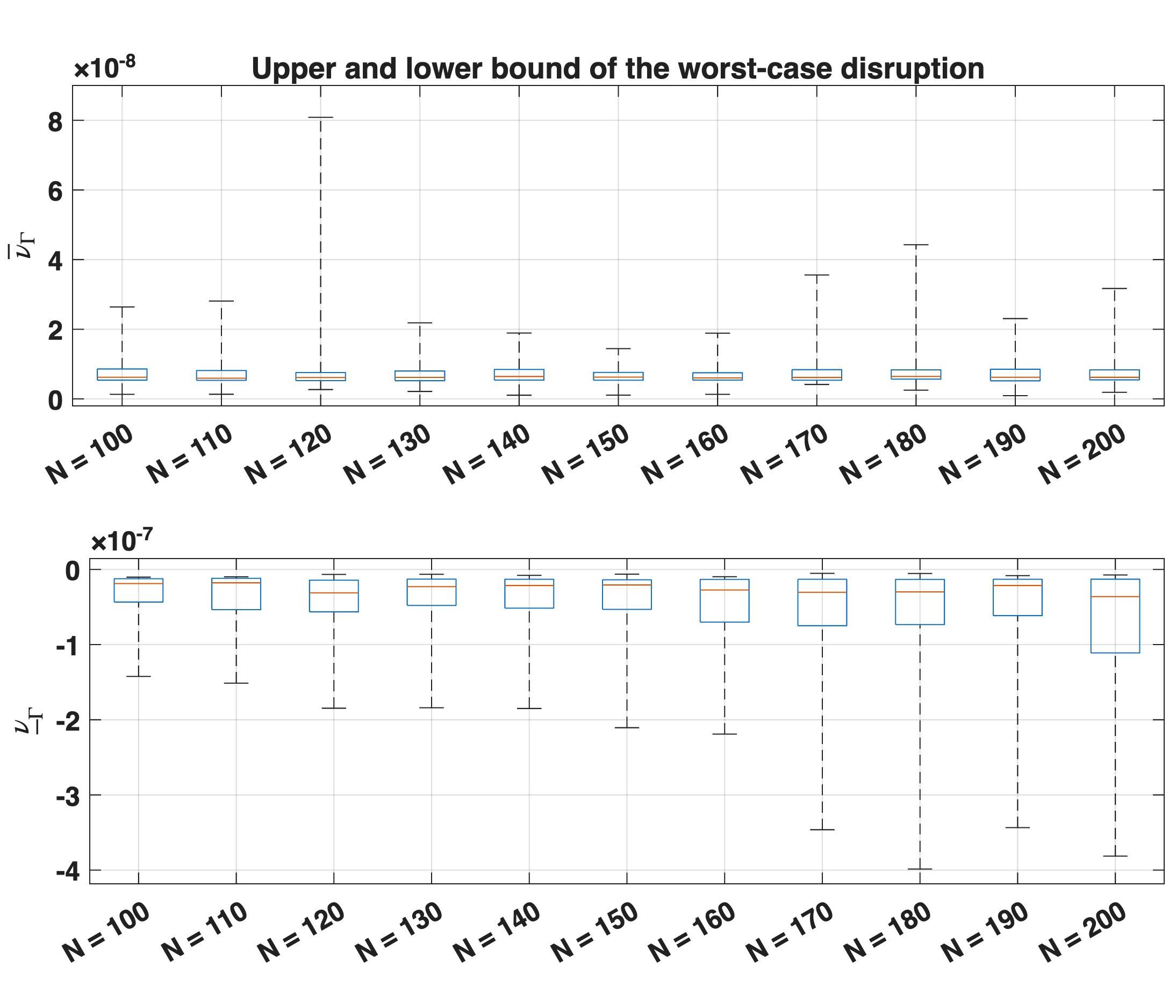}
    \caption{Relative difference between the upper bound $\Gamma(P \succeq 0, P\text{ diagonal})$ and lower bound  $\Gamma(\text{non-}P)$ for the worst-case disruption \eqref{Q_sup} across different network sizes.}
    \label{fig:Q_sup_upper_lower}
\end{figure}
\begin{figure}[!t]
    \centering
    \begin{subfigure}
        [t]{0.98\linewidth}
        \includegraphics[width=\linewidth]{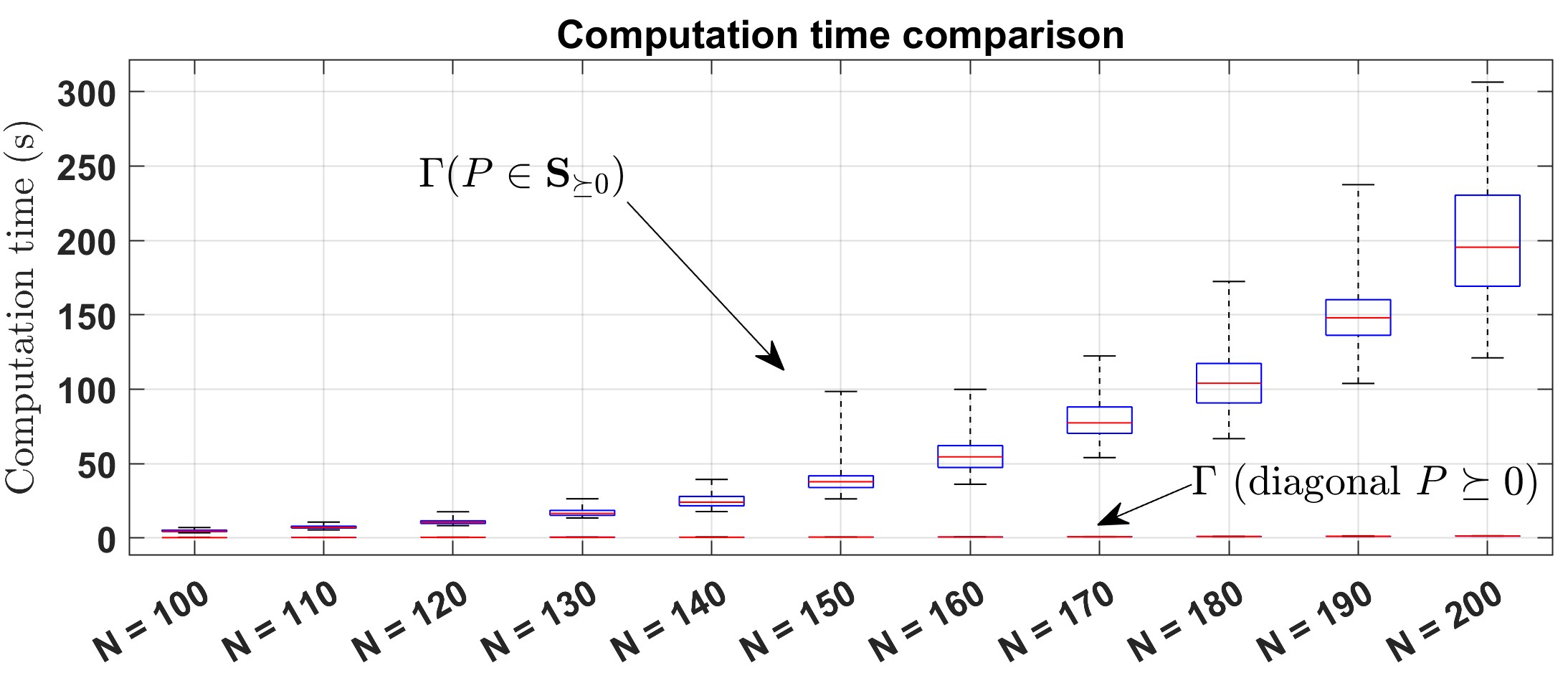}
        \caption{}
    \end{subfigure}
    \begin{subfigure}
        [t]{0.98\linewidth}
        \includegraphics[width=\linewidth]{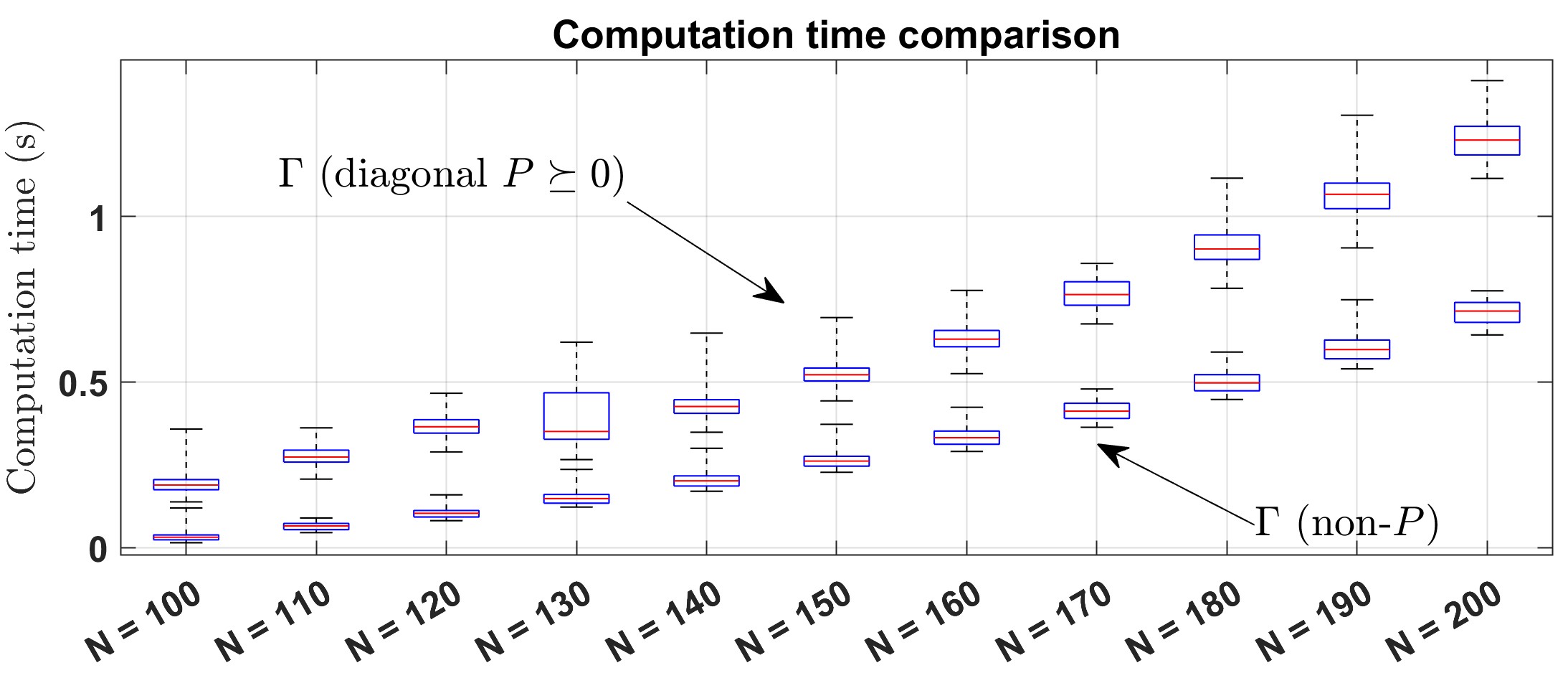}
        \caption{}
    \end{subfigure}
    \caption{
    Computation time comparison for (a) the full SDP  $\Gamma(P \in \Sbb^N_{\succeq 0})$ versus its diagonal relaxation, and  (b) the diagonal relaxation versus the non-$P$ formulation \eqref{Q_min_sdp_short}.}
    \label{fig:time_Q_sup_upper_lower}
\end{figure}
\subsection{Validation of Theorem~\ref{thm:robust_security}}
In the first part, we choose $100$ Erdős–Rényi random directed graphs, where an edge is included to connect two vertices with a probability
of $0.25$, with different network sizes $N = \{100, 110, ..., 200\}$. For each network sample, attack nodes and monitor nodes are randomly chosen as $10\%$ of the nodes and we compute $\Gamma(P \in \Sbb^N_{\succeq 0})$, $\Gamma(P \succeq 0, P ~\text{is diagonal})$, and $\Gamma(\text{non-}P)$ in \eqref{Q_min_sdp_short}. Let us denote the following notation: 
\begin{align}
    \bar \nu_\Gamma &\triangleq \frac{\Gamma(P \succeq 0, P ~\text{is diagonal}) - \Gamma(P \in \Sbb^N_{\succeq 0})}{\Gamma(P \in \Sbb^N_{\succeq 0})}, \\
    \underline \nu_\Gamma &\triangleq \frac{\Gamma(\text{non-}P) - \Gamma(P \in \Sbb^N_{\succeq 0})}{\Gamma(P \in \Sbb^N_{\succeq 0})},
\end{align}
where $\bar \nu_\Gamma$ and $\underline \nu_\Gamma$ stand for the relative upper bound and relative lower bound for $\Gamma(P \in \Sbb^N_{\succeq 0})$, respectively.
The simulation results with different network sizes are reported in Fig.s~\ref{fig:Q_sup_upper_lower}-\ref{fig:time_Q_sup_upper_lower}.  From Fig.~\ref{fig:Q_sup_upper_lower}, we observe that the relative differences between $\Gamma(P \in \Sbb^N_{\succeq 0})$, $\Gamma(P \succeq 0, P ~\text{is diagonal})$, and $\Gamma(\text{non-}P)$ are very small, less than $4 \times 10^{-5} (\%)$ across all the network sizes, which confirms the result presented in Theorem~\ref{thm:robust_security}. We also see the computation time for $\Gamma(P \succeq 0, P ~\text{is diagonal})$ and $\Gamma(\text{non-}P)$ are significantly less than that for $\Gamma(P \in \Sbb^N_{\succeq 0})$ in Fig.~\ref{fig:time_Q_sup_upper_lower}, which confirms that $\Gamma(P \succeq 0, P ~\text{is diagonal})$ and $\Gamma(\text{non-}P)$ do scale better with network sizes.

In Theorem 1, we showed that if the sufficient condition \eqref{cond:W2_Hinf} is satisfied, $\Gamma(P \in \Sbb^N_{\succeq 0})$, $\Gamma(P \succeq 0, P~\text{is diagonal})$, and $\Gamma(\text{non-}P)$ yield the same value, which has been confirmed in the previous simulation results (see Fig.~\ref{fig:Q_sup_upper_lower}). A natural question is whether $\Gamma(P \succeq 0, P~\text{is diagonal})$ and $\Gamma(\text{non-}P)$ still remain meaningful when \eqref{cond:W2_Hinf} does not hold. Since
computing $\Gamma(P \in \Sbb^N_{\succeq 0})$ can be challenging in practice due to its heavy hardware demand, we evaluate the gap between $\Gamma(P \succeq 0, P~\text{is diagonal})$ and $\Gamma(\text{non-}P)$, which serve as the upper bound and the lower bound for $\Gamma(P \in \Sbb^N_{\succeq 0})$, respectively. Fig.~\ref{fig:Q_sup_ratio} shows this gap for $100$ Erdős–Rényi random directed graphs with network sizes ranging from $N = 100$ to $N = 200$. For each network sample, attack nodes and monitor nodes are randomly chosen as $10\%$ of the nodes. Although \eqref{cond:W2_Hinf} does not hold in these experiments, the ratio remains close to $1$ across $100$ trials per network size. Moreover, the median ratio is consistently equal to $1$, the $75$-th percentile ratio varies in $[1,~1.29]$, and they do not grow with the network size. These results confirm that $\Gamma(P \succeq 0, P~\text{is diagonal})$ and $\Gamma(\text{non-}P)$ provide a meaningful estimate for the worst-case impact of stealthy attacks for large networks.

\begin{figure}[!t]
    \centering
    \includegraphics[width=\linewidth]{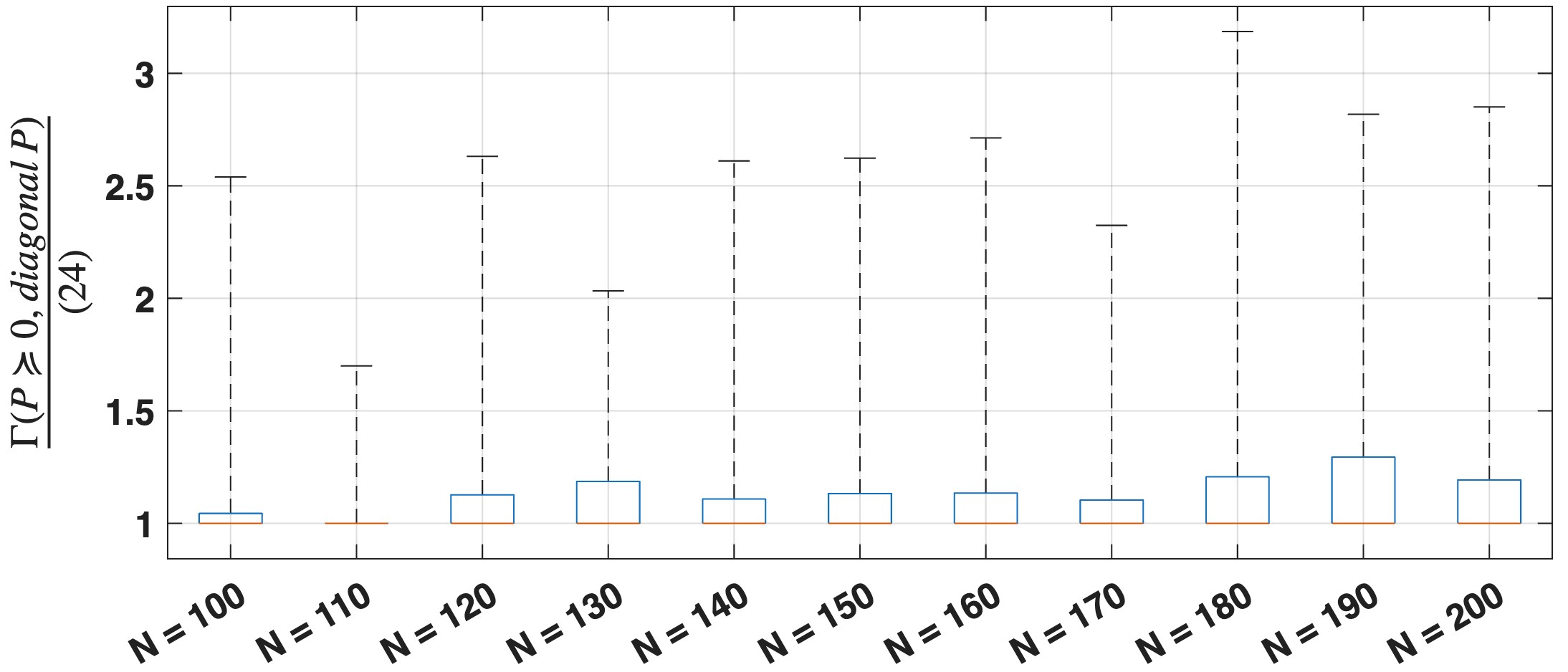}
    \caption{Ratio of the upper and lower bounds of the worst-case disruption \eqref{Q_sup} with different network sizes.}
    \label{fig:Q_sup_ratio}
\end{figure}

\subsection{Validation of Algorithm~\ref{alg:principal_monitor}}
In this part, we aim to show the efficiency of  Algorithm~\ref{alg:principal_monitor} by comparing the worst-case disruption induced by the heuristic-search monitor set, denoted by $\Mc^\star_{\text{Katz}}$, found from Algorithm~\ref{alg:principal_monitor} and the budget-constrained optimal monitor set, denoted by $\Mc^\star_{\text{MISDP}}$ found by solving a mixed-integer SDP \cite{nguyen2025scalable}, where we choose $\beta = \alpha + 1$. In the comparison, we show the following values:
\begin{align}
    \rho_Q &=  \frac{\displaystyle \max_{|\Ac| =\alpha} Q(\Mc^\star_{\text{Katz}}, \Ac) - \max_{|\Ac| =\alpha} Q(\Mc^\star_{\text{MISDP}}, \Ac)}{\displaystyle \max_{|\Ac| =\alpha} Q(\Mc^\star_{\text{MISDP}}, \Ac)}, \\
    \Delta \beta &= |\Mc^\star_{\text{Katz}}| - \beta.
\end{align}
It is worth noting that $\rho_Q$ is not negative if $|\Mc^\star_{\text{Katz}}| = \beta$ since $\Mc^\star_{\text{MISDP}}$ is the budget-constrained optimal monitor set \cite[Theorem 1]{nguyen2025scalable}. However, the rank-repair step (lines 14-15 in Algorithm~\ref{alg:principal_monitor}) allows $\Mc^\star_{\text{Katz}}$ to have more nodes than $\beta$, resulting in $\rho_Q$ being negative. As discussed in Remark~\ref{rem:prob_clarify}, the mixed-integer SDP is NP-hard and requires a large computational resource. Due to hardware limitations, we mainly perform the simulations on networks of size $N = 20$ and $N = 30$. For each network, we take $100$ Erdős–Rényi samples with a probability of $0.25$.

The simulation result is reported in Fig.~\ref{fig:heuristic_result}. For both network sizes, the median performance ratio is close to zero (see the left panel of Fig.~\ref{fig:heuristic_result}), indicating that the heuristic achieves a comparable result to the budget‑constrained optimal solution. Instances where the ratio is negative correspond to cases in which the heuristic selects more than $\beta$ monitor nodes due to the rank‑repair step, resulting in a lower worst‑case disruption than the budget‑constrained optimal solution. The right panel of Fig.~\ref{fig:heuristic_result} reports the number of additional monitor nodes selected by
Algorithm~\ref{alg:principal_monitor} beyond the prescribed budget~$\beta$.
For both network sizes, the median and the 75th percentile of $\Delta\beta$
are equal to~1, indicating that the heuristic requires only a single additional monitor node in the vast majority of cases. This behavior reflects the fact
that many random graph instances do not admit any feasible monitor set of size
$\beta$ due to the rank condition in Theorem~\ref{lem:ellisoid_opt}. In these
cases, Algorithm~1 increases the number of monitors until feasibility is
restored, which explains the occasional presence of large outliers. From a
practical perspective, deploying one extra monitor node is typically a minor increase in resource usage, especially when it prevents unbounded worst-case impact and yields a performance close to the budget-constrained optimal
selection.

\begin{figure} [!t]
    \centering
    \includegraphics[width=\linewidth]{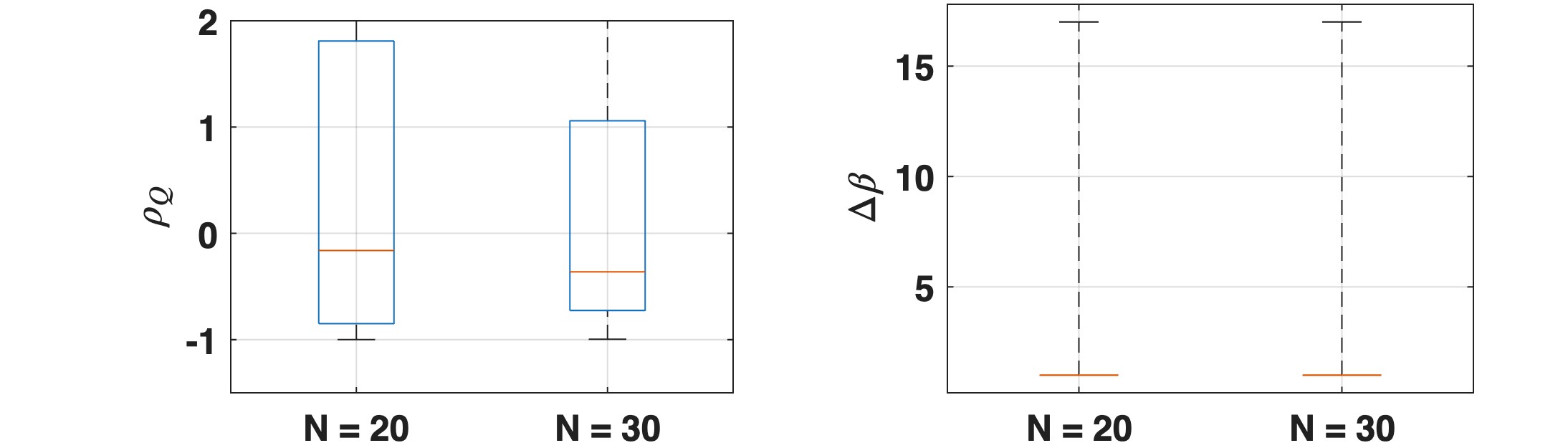}
    \caption{Comparison between the heuristic monitor set $\Mc^\star_{\text{Katz}}$ returned by Algorithm~\ref{alg:principal_monitor} and the budget-constrained optimal monitor set $\Mc^\star_{\text{MISDP}}$ obtained via the mixed-integer SDP in~\cite{nguyen2025scalable}.  (Left) performance ratio $\rho_Q$ over 100 Erdős–Rényi random graphs. (Right) number of additional monitor nodes  $\Delta\beta$ selected by Algorithm~1 beyond the prescribed budget $\beta=\alpha+1$.}
    \label{fig:heuristic_result}
\end{figure}

\section{Conclusion}
\label{sec:conclusion}
This paper studied the security of continuous‑time positive networked control systems under stealthy false data injection attacks. By leveraging the positivity of the system dynamics, we derived an SDP characterization of the worst‑case disruption and showed that, under a sufficient robustness condition, this SDP becomes exact and simplifies to a network-size-independent formulation. The resulting expression exposes a direct connection between the security metric and two adapted Katz-like centrality measures, offering a graph‑theoretic interpretation of attack impact and monitor effectiveness. When the robustness condition is not satisfied, the proposed diagonal and non‑$P$ relaxations provide computationally efficient upper and lower bounds, yielding meaningful estimates for large networks. Leveraging the geometric structure of the simplified problem, we developed a principal‑direction heuristic for monitor allocation that avoids solving optimization problems. Simulations on Erdős-Rényi networks demonstrated that the heuristic achieves sub-optimal performance while requiring only a small budget increase to restore feasibility. Overall, the proposed analysis and heuristic tools offer a scalable and practical solution for securing large positive networks against stealthy false data injection attacks.
\bibliographystyle{ieeetr}
\bibliography{mybibfile}

\appendices
\section{KYP Lemma for positive systems}
\label{app:kyplemma_positive}
We next provide a result, which is adopted from \cite[Theorem 1]{rantzer2015kalman} and \cite[Chapter 3]{moylan2014dissipative}.
\begin{Proposition}\label{thm_anders}
Let us consider a continuous LTI system $\Sigma$ with a state-space model: $\dot x(t) = A x(t) + B u(t),~y_1(t) = C_1 x(t),~y_2(t) = C_2 x(t)$ 
where $A$ is a Metzler matrix,
${B} \in \Rbb^{n \times m}_{\geq0}, {C_1} \in \Rbb^{p_1 \times n}_{\geq0},$ and ${C_2} \in \Rbb^{p_2 \times n}_{\geq0}$.
Let us consider a quadratic supply rate {$\hat s(t) = \|\Gamma^{\frac{1}{2}} u(t)\|_2^2 +  \|y_1(t)\|_2^2 - \|y_2(t)\|_2^2$ where $\Gamma$ is a positive diagonal matrix} and let $C_2^\top C_2 \geq C_1^\top C_1$ where the inequality applies element-wise. Then, the following statements are equivalent. 
\begin{enumerate}
    \item[P1)] For all trajectories of the system with $x(0)=0, x(\infty) = 0$, we have $\int_{0}^\infty \hat s(t) \text{d}t \geq 0$;
    \item[P2)] There exists a symmetric matrix $P$ such that
    \begin{equation}
    \hspace{-0.3cm}
        \begin{bmatrix}
            A^\top P + PA & PB \\ B^\top P & 0 
        \end{bmatrix} + \ba{cc}
        C_2^\top C_2 - C_1^\top C_1 & 0\\ 0 & {-\Gamma}
        \ea \preceq 0. \non
    \end{equation}
    \item[P3)] For all $\omega \in [0, \infty]$, it is true that 
    \begin{align}
        \ba{cc} 
        (i \omega I - A)^{-1} B \\ I \ea^\ast  \ba{cc}
        C_2^\top C_2 - C_1^\top C_1 & 0\\ 0 & {-\Gamma}
        \ea& \non \\
        &\hspace{-3.5cm}
        \times \ba{cc} 
        (i \omega I - A)^{-1} B \\ I \ea \preceq 0. \non 
    \end{align}
    \item[P4)] 
    \begin{align}
        \ba{cc} 
        -A^{-1} B \\ I \ea^\ast \ba{cc}
        C_2^\top C_2 - C_1^\top C_1 & 0\\ 0 & {-\Gamma}
        \ea& \non \\ 
        & \hspace{-1.5cm}
        \times \ba{cc} 
        - A^{-1} B \\ I \ea \preceq 0. \non
    \end{align}
    \item[P5)] There exists a non-negative diagonal matrix $P$ such that 
    \begin{equation}\label{LMI_anders}
    \hspace{-0.3cm}
        \begin{bmatrix}
            A^\top P + PA & PB \\ B^\top P & 0 
        \end{bmatrix} + \ba{cc}
        C_2^\top C_2 - C_1^\top C_1 & 0\\ 0 & {-\Gamma}
        \ea \preceq 0. \non
    \end{equation}
\end{enumerate}
Here, $(\cdot)^\ast$ stands for the conjugate transpose. 
\QET
\end{Proposition}
\begin{proof}
    The first three statements P1)-P3) are equivalent according to the cyclo-dissipative systems theory \cite[Chapter 3]{moylan2014dissipative}, while the last three statements P3)-P5) are equivalent according to the result in \cite[Theorem 1]{rantzer2015kalman}. Consequently, all the statements are equivalent. 
\end{proof}

\section{Proof of Theorem~\ref{thm:robust_security}}
\label{app:thm:robust_security_pf}
\begin{proof}
    We leverage a supporting result presented in Proposition~\ref{thm_anders} in Appendix~\ref{app:kyplemma_positive}, where all the assumptions of the proposition hold true under the system dynamical parameters and the satisfaction of \eqref{cond:W2_Hinf}. Note that \eqref{cond:W2_Hinf} yields \eqref{cond:in_KYP} to fulfill the condition $C_2^\top C_2 \geq C_1^\top C_1$ in Proposition~\ref{thm_anders}.
    The following part a) focuses on showing the relationship \eqref{Q_sdp_P_diagonal}. 
    
    a)
    To facilitate the proof, we consider a stricter version of \eqref{Q_sup} with a terminal state condition, i.e., $x^a(\infty) = 0$:
    \begin{align}
    \hat Q(\Mc,\Ac) \triangleq 
    ~&\sup_{\zeta \in \Rbb_{>0}^{\alpha}, \zeta \in \Lc_{2e}}~
    \norm{p^a}_{\Lc_2}^2 
    \label{Q_sup_terminal} \\
    \text{s.t.}&~
    \eqref{sys:xa}-\eqref{sys:yma}, \, x^a(0) = x^a(\infty) = 0, \non \\
    &~ \norm{y_m^a}^2_{\Lc_2} \leq \delta_m^2,~ \forall \, m \in \Mc, \non \\
    &~~ \epsilon \norm{\zeta_{a_i}}^2_{\Lc_2[0,\infty]}  \leq 1, ~\forall \, a_i \in \Ac. \non
    \end{align}
    This optimization \eqref{Q_sup_terminal} stands for the worst-case disruption on the network by a stealthy data injection attack, which lasts for a long time but in a finite time. After the attack stops, the state of the network converges to the equilibrium $x^a(\infty) = x_e = 0$ due to the asymptotic stability property. This special type of worst-case disruption supports us in 
    obtaining \eqref{Q_sdp_P_diagonal} by showing the following relations:
    \begin{align}
        &\Gamma(P \in \Sbb^N_{\succeq 0}) \overset{(i)}{\leq} \Gamma(P \succeq 0, P ~\text{is diagonal}) \overset{(ii)}{=} \Gamma(P \in \Sbb^N) \non \\
         \overset{(iii)}{=} & \hat Q(\Mc,\Ac)
        \overset{(iv)}{\leq} Q(\Mc, \Ac) \overset{(v)}{\leq}  \Gamma(P \in \Sbb^N_{\succeq 0}), \label{Q_Gamma_all_relation}
    \end{align}
    where
    \begin{itemize}
        \item $(i)$ is obtained since $\Gamma(P \succeq 0, P ~\text{is diagonal})$ has a stricter constraint on $P$ compared to the constraint on $P$ in $\Gamma(P \in \Sbb^N_{\succeq 0})$, given that they minimize the same objective function. 
        \item $(ii)$ is obtained from the equivalence between the statements P2) and P5) in Proposition~\ref{thm_anders}.
        \item $(iii)$ is obtained from the cyclo-dissipative systems theory \cite{moylan2014dissipative} and lossless S-Procedure \cite[Chapter 4]{petersen2000robust} (see \cite[Theorem 4.4]{anand2024risk} for more detailed proof).
        \item $(iv)$ is obtained since $\hat Q(\Mc,\Ac)$ is constructed by adding one more constraint to $Q(\Mc, \Ac)$ in a maximization problem.
        \item $(v)$ is obtained from \eqref{Q_upper_bound}.
    \end{itemize}
    Clearly, \eqref{Q_Gamma_all_relation} results in all the values being equal, leading to \eqref{Q_sdp_P_diagonal}. It is worth noting that the values of three types of matrix $P$ in \eqref{Q_Gamma_all_relation} are not necessarily equal. The following part b) is dedicated to showing \eqref{Q_sdp_Ainverse}.

    b) The result of \eqref{Q_sdp_P_diagonal} enables us to leverage the equivalence between the statements P4) and P5) in Proposition~\ref{thm_anders}. Consequently, $\Gamma(P \succeq 0, P ~\text{is diagonal})$ can be rewritten as follows:
    \begin{subequations} \label{Q_Ainverse_org}
    \begin{align} 
        &\Gamma(P \succeq 0, P ~\text{is diagonal}) = \min_{\gamma_m \in \Rbb_{\geq 0}, \psi \in \Rbb_{\geq 0}^\alpha}  \sum_{m \in \Mc} \gamma_m \delta_m^2 +  \textbf{1}^\top \psi \\ 
     &\text{s.t.}~~ 
     E_\Ac^\top A^\top L^{-\top} W^2 L^{-1} A E_\Ac - \epsilon \, \text{diag}(\psi) \preceq 
     \non \\
     & \hspace{1cm} 
     E_\Ac^\top A^\top L^{-\top} \bigg[ \sum_{m \in \Mc} \gamma_m \textbf{diag}\big(e_m \big) \bigg] L^{-1}A E_\Ac. \label{Q_Ainverse_cons}
    \end{align}
    \end{subequations}

    On the other hand, the Katz-like centrality measures \eqref{Katz_monitor}-\eqref{Katz_performance} can alternatively computed as follows:
    \begin{align}
    K_\delta &= \textbf{diag}(\delta)^{-1} \sum_{i = 0}^\infty ~ (D_{\rm in}^{-1} A)^{i} ~~ D_{\rm in}^{-1} A \non \\
    &= \textbf{diag}(\delta)^{-1}(I - D_{\rm in}^{-1}A)^{-1}D_{\rm in}^{-1} A \non \\
    &= \textbf{diag}(\delta)^{-1}(D_{\rm in} - A)^{-1} A \non \\
    &= \big( L \, \textbf{diag}(\delta) \big)^{-1} A, \label{Katz_monitor_mod}  
    \\
    K_W &=   
    W \sum_{i=0}^\infty (D_{\rm in}^{-1}A)^i ~~ D_{\rm in}^{-1}A
      = W L^{-1} A, \label{Katz_performance_mod}
    \end{align}
    where the second equality in \eqref{Katz_monitor_mod} occurs due to $(I - D_{\rm in}^{-1}A) \sum_{i = 0}^\infty ~ (D_{\rm in}^{-1} A)^{i} = I - D_{\rm in}^{-1}A + D_{\rm in}^{-1}A - (D_{\rm in}^{-1}A)^2 + (D_{\rm in}^{-1}A)^2 - (D_{\rm in}^{-1}A)^3 + ... = I$ and the last equality in \eqref{Katz_monitor_mod} comes from the definition of the in-degree Laplacian matrix, i.e., $L = D_{\rm in} - A$. The same derivation applies to \eqref{Katz_performance} to obtain \eqref{Katz_performance_mod}.
    It is worth noting that $-L$ is Hurwitz, so it is invertible.  
    By substituting \eqref{Katz_monitor_mod}-\eqref{Katz_performance_mod} into \eqref{Q_Ainverse_cons}, one obtains \eqref{Q_Ainverse_org} equivalent to \eqref{Q_min_sdp_short}, concluding the proof.
\end{proof}

\end{document}

%% file: deflatex3.tex
\def\Ac{{\mathcal A}}

\def\Ec{{\mathcal E}}

\def\Gc{{\mathcal G}}

\def\Hc{{\mathcal H}}

\def\Lc{{\mathcal L}}

\def\Mc{{\mathcal M}}

\def\Nc{{\mathcal N}}

\def\Rbb{{\mathbb R}}

\def\Sbb{{\mathbb S}}

\def\Vc{{\mathcal V}}

\def\0{{\bf 0}}

\newcommand{\bitem}{\begin{itemize}}
\newcommand{\eitem}{\end{itemize}}
\newcommand{\btabular}{\begin{tabular}}
\newcommand{\etabular}{\end{tabular}}
\newcommand{\bcenter}{\begin{center}}
\newcommand{\ecenter}{\end{center}}
\newcommand{\bea}{\begin{eqnarray}}
\newcommand{\eea}{\end{eqnarray}}
\newcommand{\bean}{\begin{eqnarray*}}
\newcommand{\eean}{\end{eqnarray*}}
\newcommand{\ba}{\left[ \begin{array}}
\newcommand{\ea}{\\ \end{array} \right]}
\newcommand{\bear}{\begin{array}}
\newcommand{\eear}{\\ \end{array}}

\newcommand{\non}{\nonumber}

\newcommand{\Lra}{\Leftrightarrow}

\newcommand*{\QET}{\hfill\ensuremath{\triangleleft}}
\newcommand{\norm}[1]{\left\lVert#1\right\rVert}

\newcounter{subequation}
\def\beasub{\addtocounter{equation}{+1}
\setcounter{subequation}{\value{equation}}
\setcounter{equation}{0}
\renewcommand{\theequation}{\arabic{subequation}\alph{equation}}
\begin{eqnarray}}
\def\eeasub{\end{eqnarray}
\setcounter{equation}{\value{subequation}}
\renewcommand{\theequation}{\arabic{equation}}}

